\documentclass{theoretics}

\title{Query-Efficient Fixpoints of \texorpdfstring{$\ell_p$}{lp}-Contractions}

\ThCSauthor[aff1]{Sebastian Haslebacher}{sebastian.haslebacher@inf.ethz.ch}[0000-0003-3988-3325]
\ThCSauthor[aff1]{Jonas Lill}{jolill@ethz.ch}[0009-0007-1448-2425]
\ThCSauthor[aff1,aff2]{Patrick Schnider}{patrick.schnider@inf.ethz.ch}[0000-0002-2172-9285]
\ThCSauthor[aff1]{Simon Weber}{simon.weber@inf.ethz.ch}[0000-0003-1901-3621]
\ThCSaffil[aff1]{Department of Computer Science, ETH Zürich, Switzerland}
\ThCSaffil[aff2]{Department of Mathematics and Computer Science, University of Basel, Switzerland}
\ThCSthanks{A preliminary version of this article appeared at FOCS 2025~\cite{DBLP:conf/focs/HaslebacherLSW25}.}
\ThCSshortnames{S.~Haslebacher, J.~Lill, P.~Schnider, S.~Weber}
\ThCSshorttitle{Query-Efficient Fixpoints of \texorpdfstring{$\ell_p$}{lp}-Contractions}
\ThCSyear{2026}
\ThCSarticlenum{17}
\ThCSreceived{Jan 21, 2026}
\ThCSrevised{Aug 4, 2026}
\ThCSaccepted{Aug 17, 2026}
\ThCSpublished{Oct 27, 2026}
\ThCSdoicreatedtrue
\ThCSkeywords{Contraction Map, \texorpdfstring{$\ell_p$}{lp}-Contraction, Fixpoint, Fixed Point, Centerpoint}

\usepackage{comment}

\usepackage{mathtools}

\usepackage{aligned-overset}

\usepackage[capitalise,noabbrev,nameinlink]{cleveref}
\usepackage{thm-restate}
\usepackage{nicefrac}

\makeatletter
\newcommand\IfRestateTF{%
  \ifx\label\thmt@gobble@label 
    \expandafter\@firstoftwo
  \else
    \expandafter\@secondoftwo
  \fi
}
\makeatother
\newcommand{\RestateRemark}{\IfRestateTF{{\normalfont\bfseries (Restated) }}{}}

\crefname{corollary}{Corollary}{Corollaries}
\crefname{lemma}{Lemma}{Lemmata}
\crefname{observation}{Observation}{Observations}
\crefname{cfigure}{Figure}{Figures}

\newcommand{\NP}{\ensuremath{\mathsf{NP}}\xspace}

\newcommand{\CoNP}{\ensuremath{\mathsf{CoNP}}\xspace}

\renewcommand{\P}{\ensuremath{\mathsf{P}}\xspace}
\newcommand{\CLS}{\ensuremath{\mathsf{CLS}}\xspace}
\newcommand{\PLS}{\ensuremath{\mathsf{PLS}}\xspace}
\newcommand{\EOPL}{\ensuremath{\mathsf{EOPL}}\xspace}
\newcommand{\UEOPL}{\ensuremath{\mathsf{UEOPL}}\xspace}
\newcommand{\PPAD}{\ensuremath{\mathsf{PPAD}}\xspace}

\newcommand{\FPdt}{\ensuremath{\mathsf{FP}^{\text{dt}}}\xspace}
\newcommand{\TFNPdt}{\ensuremath{\mathsf{TFNP}^{\text{dt}}}\xspace}

\newcommand{\N}{\mathbb{N}}

\newcommand{\R}{\mathbb{R}}

\newcommand{\bigO}{\mathcal{O}}
\DeclareMathOperator{\poly}{poly}

\DeclareMathOperator{\vol}{vol}

\DeclarePairedDelimiter\norm{\lVert}{\rVert}

\let\emptyset\varnothing

\renewcommand{\epsilon}{\ensuremath\varepsilon}

\renewcommand{\phi}{\ensuremath{\varphi}}

\newcommand{\conv}{\text{conv}}

\newcommand{\limitH}{\ensuremath{\mathcal{H}}}
\newcommand{\bisecH}{\ensuremath{H}}

\newcommand{\continuousProblem}[1][p]{\textsc{$\ell_{#1}$-ContractionFixpoint}\xspace}
\newcommand{\gridToGrid}[1]{\textsc{$\ell_{#1}$-GridContractionFixpoint}\xspace}
\newcommand{\totalContraction}[1]{\textsc{Total-$\ell_{#1}$-ContractionFixpoint}\xspace}

\begin{document}

\maketitle

\begin{abstract}
    We prove that an $\epsilon$-approximate fixpoint of a map $f:[0,1]^d\rightarrow [0,1]^d$ can be found with $\bigO(d^2(\log\frac{1}{\epsilon} + \log\frac{1}{1-\lambda}))$ queries to $f$ if $f$ is $\lambda$-contracting with respect to an $\ell_p$-metric for some $p\in [1,\infty)\cup\{\infty\}$. This generalizes a recent result of Chen, Li, and Yannakakis~[STOC'24] from the $\ell_\infty$-case to all $\ell_p$-metrics. Previously, all query upper bounds for $p\in [1,\infty) \setminus \{2\}$ were either exponential in $d$, $\log\frac{1}{\epsilon}$, or $\log\frac{1}{1-\lambda}$.

    Chen, Li, and Yannakakis also show how to ensure that all queries to $f$ lie on a discrete grid of limited granularity in the $\ell_\infty$-case. We provide such a rounding for the $\ell_1$-case, placing an appropriately defined version of the $\ell_1$-case in \FPdt.

    To prove our results, we introduce the notion of $\ell_p$-halfspaces and generalize the classical centerpoint theorem from discrete geometry: for any $p \in [1, \infty) \cup \{\infty\}$ and any mass distribution (or point set), we prove that there exists a centerpoint $c$ such that every $\ell_p$-halfspace defined by $c$ and a normal vector contains at least a $\frac{1}{d+1}$-fraction of the mass (or points).
\end{abstract}

\section{Introduction}

A \emph{contraction map} is a function that maps any two points in such a way that their two images lie closer to each other than the original points. Formally, a function~$f:X\rightarrow X$ on a metric space~$(X,d_X)$ is a contraction map if there exists some $\lambda<1$, such that any two points $x,y\in X$ satisfy $d_X(f(x),f(y))\leq \lambda\cdot d_X(x,y)$. Banach's fixpoint theorem~\cite{banach1922operations} famously states that a contraction map must have a unique fixpoint, i.e., a point $x \in X$ with $f(x)=x$.

In this paper, we consider the problem of finding \emph{approximate} fixpoints based on the \emph{residual error criterion}: we wish to find any point with $d_X(f(x),x)\leq \epsilon$, called an $\epsilon$-approximate fixpoint. Banach's proof actually provides a simple algorithm to find such an $\epsilon$-approximate fixpoint: start at any point $x \in X$ and iteratively apply $f$ until $d_X(f^{k+1}(x),f^{k}(x))\leq \epsilon$. This algorithm requires at most $\bigO\left( \frac{\log(\frac{1}{\epsilon})}{\log(\frac{1}{\lambda})}\right)$ queries to $f$ on metric spaces of constant diameter.

Unfortunately, Banach's iterative algorithm is quite slow when $\lambda$ is close to $1$, which is often the regime that is useful in practice~\cite{sikorskiComputationalComplexityFixed2009}. For example, Simple Stochastic Games (SSG) reduce to approximating the fixpoint of a contraction map on the metric space induced on the unit cube $[0,1]^d$ by the $\ell_\infty$-norm, with $1 - \lambda$ and $\epsilon$ exponentially small in the size of the game, and $d$ linear in the number of states of the game~\cite{condonComplexityStochasticGames1992}. Similarly, the ARRIVAL problem reduces to approximating the fixpoint of a contraction map on the unit cube $[0,1]^d$, again with $1 - \lambda$ and $\epsilon$ exponentially small and $d$ linear in the number of vertices of the ARRIVAL instance, but here the metric is induced by the $\ell_1$-norm instead~\cite{haslebacherARRIVALRecursiveFramework2025}. Note that no polynomial-time algorithms are known for SSG and ARRIVAL, despite both of them being contained in $\NP \cap \CoNP$~\cite{condonComplexityStochasticGames1992,dohrauARRIVALZeroPlayerGraph2017}. If we could approximate the fixpoint of a $\lambda$-contraction map on $[0,1]^d$ with respect to an $\ell_p$-norm in time polynomial in $d$, $\log(\frac{1}{\epsilon})$, and $\log(\frac{1}{1-\lambda})$, then this would put both ARRIVAL and SSG in \P.

In the Euclidean case ($p=2$), such efficient algorithms have been known for quite some time~\cite{huangApproximatingFixedPoints1999, sikorskiEllipsoidAlgorithmComputation1993, sikorskiComputationalComplexityFixed2009}. Concretely, the Inscribed Ellipsoid algorithm~\cite{huangApproximatingFixedPoints1999} requires $\bigO(d\log(\frac{1}{\epsilon}))$ queries (independent of $\lambda$) to find an $\epsilon$-approximate fixpoint. Moreover, each query point can be found efficiently, implying an overall time-efficient (and not just query-efficient) algorithm.

In contrast, only little is known for $p \neq 2$. The problem is known to be contained in \CLS~\cite{daskalakisContinuousLocalSearch2011} and was explicitly mentioned as an open problem by Fearnley, Goldberg, Hollender, and Savani~\cite{fearnleyComplexityGradientDescent2022} when they proved $\CLS=\PPAD\cap\EOPL$. For 
$p\in\N\cup\{\infty\}$, Fearnley, Gordon, Mehta, and Savani proved containment in \UEOPL\ by reduction to One-Permutation Discrete Contraction (OPDC)~\cite{fearnleyUniqueEndPotential2020}.  Moreover, their algorithm for OPDC implies an algorithm for finding an $\epsilon$-approximate fixpoint in $\bigO(\log^d(\frac{1}{\epsilon}))$ queries. Unfortunately, this is still exponential in $d$. 

Chen, Li, and Yannakakis~\cite{chenComputingFixedPoint2024} recently gave the first query-efficient algorithm in the special case of the $\ell_\infty$-norm. Concretely, their algorithm can find an $\epsilon$-approximate fixpoint with $\bigO(d\log\frac{1}{\epsilon})$ (independent of $\lambda$) queries to $f$. One of the main questions that they ask is whether similar query upper bounds can be shown for $p \neq \{2, \infty\}$ as well. We answer this question affirmatively by achieving a query upper bound of $\bigO(d^2(\log\frac{1}{\epsilon}+\log\frac{1}{1-\lambda}))$ for all $p \in [1, \infty) \cup \{\infty\}$ (note that an observation of Chen, Li, and Yannakakis~\cite{chenComputingFixedPoint2024} is that one can always get rid of the factor $\log\frac{1}{1-\lambda}$ for $p = \infty$).

Neither our algorithm nor the algorithm of Chen, Li, and Yannakakis~\cite{chenComputingFixedPoint2024} is time-efficient, i.e., in both cases it is unclear how to algorithmically find the query points efficiently. However, Chen, Li, and Yannakakis~\cite{chenComputingFixedPoint2024} ensure that all queries made by their algorithm lie on a discrete grid of limited granularity, which is certainly a prerequisite for any time-efficient algorithms. We provide a similar rounding of queries to a grid in the $\ell_1$-case. 

Note that since our algorithm is only query-efficient, it does not imply polynomial-time algorithms for SSG, ARRIVAL, or other applications. However, as pointed out by Chen, Li, and Yannakakis~\cite{chenComputingFixedPoint2024}, the existence of query-efficient algorithms could be viewed as evidence for the existence of time-efficient algorithms for $\ell_p$-contraction (at least for $p \in \{1, \infty\}$, where we know how to restrict the queries to a discrete grid of limited granularity).

\subsection{Results}

For any $\ell_p$-norm with $p \in [1, \infty) \cup \{\infty\}$, we present an algorithm that can find an $\epsilon$-approximate fixpoint of a $\lambda$-contracting function $f:[0,1]^d\rightarrow [0,1]^d$ by making at most $\bigO(d^2(\log\frac{1}{\epsilon}+\log\frac{1}{1-\lambda}))$ queries to $f$. 
In the $\ell_1$-case, we show how to ensure that all queries are made on points of a given discrete grid, as long as the distance between two neighboring grid points is at most $ \frac{\epsilon}{2d} \frac{1 - \lambda}{1 + \lambda} $. A similar rounding was provided by Chen, Li, and Yannakakis~\cite{chenComputingFixedPoint2024} in the $\ell_\infty$-case. If the function $f$ produces only $\poly(k)$ bits of output when given $k$ bits of input (as it is the case with many applications of the $\ell_1$-case and the $\ell_\infty$-case), our algorithm only needs to input and output polynomially many bits of information into the function $f$.

As a main technical tool, we consider a generalization of the notion of halfspaces from Euclidean geometry into $\ell_p$-geometry. In particular, we prove a generalization of the classical \emph{centerpoint theorem} from discrete and computational geometry that works for any $p \in [1, \infty) \cup \{\infty\}$. We think that this might be of independent interest. Our generalized centerpoint theorem works for both mass distributions and point sets and is shown to be tight in both cases for $p \in (1, \infty)$.

Our results also imply that a properly defined total search problem version of the $\ell_1$-case lies in \FPdt, the class of black-box total search problems that can be solved efficiently by decision trees~\cite{goosSeparationsProofComplexity2024}. Again, the analogous observation for the $\ell_\infty$-case was made by Chen, Li, and Yannakakis~\cite{chenComputingFixedPoint2024} already. In a total search problem, the algorithm has to either return a solution (an $\epsilon$-approximate fixpoint) or a proof that the given instance does not fulfill the promise (i.e., that the given function on the grid does not extend to a contraction map on all of $[0,1]^d$). In the $\ell_\infty$-case, it suffices to consider pairs of grid points violating the contraction property as proof for the violation of the promise~\cite{chenComputingFixedPoint2024}. For the $\ell_1$-case, this does not seem to suffice, which means that we need to consider a slightly more complicated proof of violation to properly define a total search version of the problem.

\subsection{Proof Techniques}

Our algorithm is based on a simple observation that has also been used in previous work~\cite{chenComputingFixedPoint2024,sikorskiEllipsoidAlgorithmComputation1993}: due to the contraction property, given any query point $x$, the fixpoint $x^\star$ must be closer to the query's response $f(x)$ than to the query $x$ itself. To quickly hone in on the fixpoint, we therefore wish to query a point $x$ such that for every possible response $f(x)$ of $f$, a significantly large part of the remaining search space lies at least as close to $x$ as to $f(x)$. Querying $x$ then allows us to considerably shrink the remaining search space no matter the response $f(x)$.

In the Euclidean case ($\ell_2$-metric), a \emph{centerpoint} of the remaining search space makes for a good query: a point $c$ is a $\rho$-centerpoint of a mass distribution $\mu$ if any halfspace that contains $c$ also contains at least a $\rho$-fraction of the mass of $\mu$. The celebrated centerpoint theorem originally due to Rado~\cite{radoTheoremGeneralMeasure1946} guarantees that a $\frac{1}{d+1}$-centerpoint (commonly just called a centerpoint) always exists. Since the set of points that are at least as close to $c$ as to $f(c)$ is a halfspace containing $c$, querying a centerpoint $c$ of the remaining search space allows us to discard at least a $\frac{1}{d+1}$-fraction of the search space.

We generalize these concepts from Euclidean geometry to arbitrary $\ell_p$-metrics. We first define \emph{$\ell_p$-halfspaces} as a generalization of halfspaces. Informally speaking, an $\ell_p$-halfspace $\limitH^p_{x, v}$ is defined by a point $x \in \R^d$ and a direction $v \in S^{d -1}$ and contains all points that are at least as close (with respect to the $\ell_p$-metric) to $x$ as to $x - \epsilon v$ for all $\epsilon > 0$ (see \Cref{fig:limithyperplaneExample}). We then prove a generalized centerpoint theorem for $\ell_p$-metrics, which says that for every mass distribution $\mu$ on $\R^d$, there exists a point $c$ satisfying $\mu(\limitH^p_{c, v}) \geq \frac{1}{d + 1} \mu(\R^d)$ for all $v \in S^{d - 1}$. This theorem implies existence of good query points for our algorithm, allowing us to discard a $\frac{1}{d+1}$-fraction of the remaining search space with each query. 

\begin{figure}[ht]
    \centering
    \includegraphics[width=0.6\linewidth]{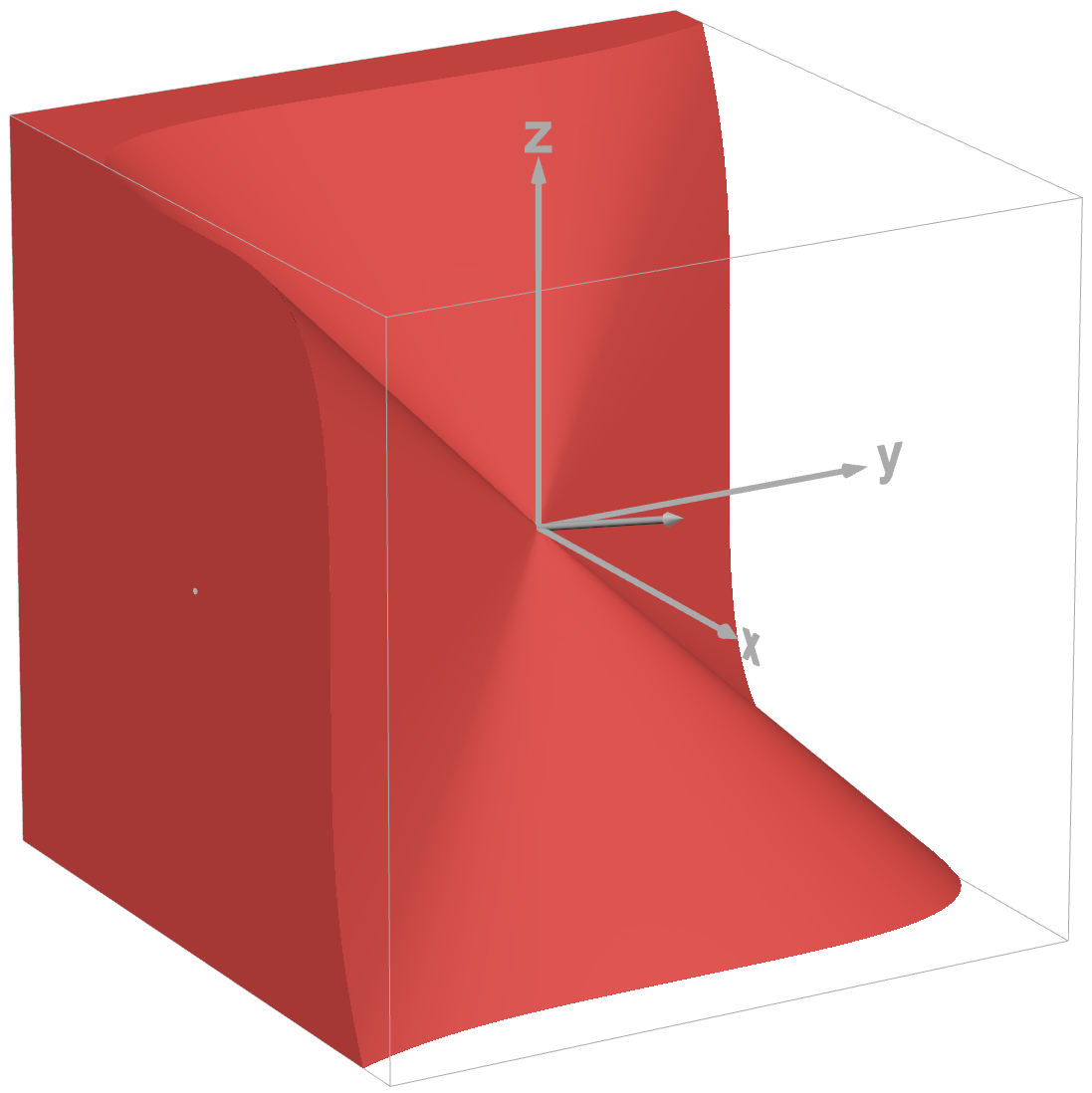}
    \caption{The $\ell_5$-halfspace $\limitH^{5}_{\mathbf{0},v}$ for $v=(-2,-0.5,-0.6)$ is drawn in red. The vector shown in the image is $-v$. Image created with the Desmos 3D calculator.}
    \label{fig:limithyperplaneExample}
\end{figure}

To prove our generalized centerpoint theorem, we use Brouwer's fixpoint theorem~\cite{brouwerUeberAbbildungMannigfaltigkeiten1911}. Concretely, given a mass distribution $\mu$ on $\R^d$, we consider a function $F : \R^d \rightarrow \R^d$ which intuitively maps any point $x$ by pushing it in all directions $v \in S^{d - 1}$ for which $\mu(\limitH^p_{x, -v})$ is not yet large enough. Formally, we define it using the following integral
\[
    F_i(x) \coloneqq x_i + \int_{S^{d - 1}} v_i \max \left( \frac{1}{d + 1} - \mu(\limitH^p_{x, -v}), 0 \right) dv
\]
for all $i \in [d]$. We show that $F$ is continuous, which means that if we restrict it to a compact domain, Brouwer's fixpoint theorem implies existence of a fixpoint. We then prove that this fixpoint must be a centerpoint. 

After proving that we can reduce our remaining search space by a $\frac{1}{d+1}$-fraction with every query, we still need a termination condition for our algorithm. To this end, we show that as long as a queried point $x$ is not an $\epsilon$-approximate fixpoint, any point sufficiently close to the fixpoint $x^\star$ cannot lie closer to $x$ than to $f(x)$. Thus, there exists a small ball around $x^\star$ that cannot be discarded as long as we have not queried an $\epsilon$-approximate fixpoint yet. This in turn implies a lower bound $V_{\epsilon, \lambda}$ for the volume of the remaining search space that holds as long as the algorithm does not terminate. From this, we can then conclude the upper bound 
\[
    \log_{\frac{d}{d+1}} \left(\frac{1}{V_{\epsilon, \lambda}} \right) = \bigO \left( d^2 \left(\log \frac{1}{\epsilon} + \log \frac{1}{1 - \lambda} + \log d \right) \right)
\]
on the number of queries needed to find an $\epsilon$-approximate fixpoint. To get rid of the $\log d$ term, we additionally observe that a simple iterative algorithm is faster if $d$ dominates $\frac{1}{\epsilon}$ and $\frac{1}{1 - \lambda}$.

As mentioned before, Chen, Li, and Yannakakis~\cite{chenComputingFixedPoint2024} provide a rounding scheme in the $\ell_\infty$-case that allows them to ensure that all queries lie on a discrete grid of limited granularity. Unfortunately, their rounding scheme is specific to the $\ell_\infty$-case. For the $\ell_1$-case, which is also motivated by applications, we therefore provide our own rounding scheme. Concretely, we show that for $p = 1$, each centerpoint query $c$ of the original algorithm can be replaced with a query to the closest grid point $c'$ to $c$. To make this work, we need to use a discrete variant of our centerpoint theorem to measure the size of the remaining search space in terms of remaining grid points.

\subsection{Discussion}

\paragraph{Comparison to Results for $p=2$.}
The intersection of Euclidean halfspaces is a polyhedron, and Grünbaum's theorem says that the centroid of a polyhedron is a $\frac{1}{e}$-centerpoint of the polyhedron~\cite{grunbaumPartitionsMassdistributionsConvex1960}. Thus, there always exists a $\frac{1}{e}$-centerpoint of the remaining search space in the $\ell_2$-case, rather than just a $\frac{1}{d+1}$-centerpoint. Thanks to this, specialized algorithms such as the Inscribed Ellipsoid and the Centroid algorithm~\cite{sikorskiEllipsoidAlgorithmComputation1993,sikorskiComputationalComplexityFixed2009} achieve a query upper bound that depends only linearly on $d$ (compared to our quadratic dependency). An interesting question would be whether something analogous to Grünbaum's theorem exists for the intersection of $\ell_p$-halfspaces in general. Such a bound could then be plugged into our analysis to improve our upper bounds from $\bigO(d^2)$ to $\bigO(d\log d)$ in terms of the dimension. Note that the $\log d$ overhead (compared to aforementioned results with $\bigO(d)$ dependency for the $\ell_2$-case) would come from using a lower bound of $\vol(B^p(0,b))\geq \frac{2^d}{d!}b^d$ for the volume of a general $\ell_p$-ball in our analysis (see proof of \Cref{thm:mainContinuous} for more details).

As we have discussed earlier, the Inscribed Ellipsoid algorithm is not just query-efficient, but can actually be implemented in polynomial time. However, this is not straightforward. Even in the Euclidean case, computing the centroid of a polyhedron (given by its bounding hyperplanes) is $\mathsf{\#P}$-hard~\cite{rademacherApproximatingCentroidHard2007}. 
One thus has to rely on approximate centerpoints that can be computed efficiently by taking the centroid of an ellipsoid that approximates the polyhedron (hence the name Inscribed Ellipsoid algorithm).
In other words, using ellipsoids to approximate the remaining search space is the main tool to get from query-efficient algorithms to actual polynomial-time algorithms in the Euclidean case. 

For more information regarding the $\ell_2$-case, we refer to the survey by Sikorski~\cite{sikorskiComputationalComplexityFixed2009}.

\paragraph{Comparison to Chen, Li, and Yannakakis~\cite{chenComputingFixedPoint2024} on $p=\infty$.}
Our approach is similar (but much more general, albeit weaker in the case $p = \infty$) to the one of Chen, Li, and Yannakakis~\cite{chenComputingFixedPoint2024} for the $\ell_\infty$-case. They prove that there always exists a query point $x$, such that any response $f(x)$ will allow them to discard at least a $\frac{1}{2}$-fraction of the remaining search space. Concretely, the area that is discarded consists of $d$ out of $2d$ \emph{pyramids} ($d$ pairs of antipodal pyramids) that touch at $x$. To prove existence of a good query $x$, they use Brouwer's fixpoint theorem on a function that balances the mass among the two pyramids in each of the $d$ pairs of antipodal pyramids simultaneously. In other words, their good query point $x$ is essentially a $\frac{1}{2}$-centerpoint, although they do not use this terminology.
Their decomposition into pyramids is tailored to the $\ell_\infty$-case and does not generalize. 

Our algorithm for arbitrary $p$ uses $\bigO(d^2(\log\frac{1}{\epsilon}+\log\frac{1}{1-\lambda}))$ queries, which means that we pay an additional linear factor in $d$ compared to~\cite{chenComputingFixedPoint2024}.Note that their upper bound does not include $\lambda$ because they show how to get rid of the dependency on $\lambda$ by applying a reduction that turns an instance with parameters $\epsilon$ and $\lambda$ into an instance with $\epsilon'=\epsilon/2$ and $\lambda'=1-\epsilon'$. This is again specific to the case $p=\infty$ and unlikely to work in general. 

In order to prove the result for general $p$, we avoid a specific decomposition such as the pyramid decomposition in the $\ell_\infty$-case. We think that the generality of our approach also makes our arguments less technical than the ones used in~\cite{chenComputingFixedPoint2024}. 

Chen, Li, and Yannakakis~\cite{chenComputingFixedPoint2024} show that their algorithm also places the following total version of the problem in \FPdt: given query access to the function $f$ on a grid, either find an $\epsilon$-approximate fixpoint, or two grid points violating the contraction property. Totality of this problem crucially relies on the fact that any $\ell_\infty$-contracting function defined only on the grid can be extended to an $\ell_\infty$-contracting map on the whole cube. Unfortunately, their extension construction, which is an implicit application of the extension theorem of McShane~\cite{mcshaneExtensionRangeFunctions1934}, does not generalize to other values of $p$. In fact, we do not expect a similar extension theorem to hold for general $p$. In order to define a total version of the problem in the $\ell_1$-case, we therefore need to include a different certificate for violation of the contraction property. This certificate uses not just two, but polynomially many grid points, and exists whenever our algorithm would fail to find an $\epsilon$-approximate fixpoint (on a non-contracting map).

\paragraph{\CLS-Completeness of General Banach.}
Daskalakis, Tzamos, and Zampetakis~\cite{daskalakisConverseBanachFixed2018} showed that the problem of computing the fixpoint of a contraction map on the metric space $([0,1]^3, d_{[0,1]^3})$ is \CLS-complete if the metric $d_{[0,1]^3}$ is an arbitrary metric given as an algebraic circuit. This stands in stark contrast to our polynomial-query result for all $\ell_p$-metrics,
and highlights that the $\ell_p$-metrics are quite special metrics. For example, as metrics induced by a norm, the $\ell_p$-metrics are translation-invariant. In future work, it would be interesting to investigate whether super-polynomial query lower bounds can be proven for some fixed metric, or whether our results can be generalized to more general classes of metrics, for example all metrics induced by some norm.

\paragraph{Finding Centerpoints.}
An important question left unanswered by our work is whether our algorithms can be implemented not only using polynomially many queries, but also in overall polynomial time. Such a result would have important implications. For example, it would imply the existence of a polynomial-time algorithm for SSGs~\cite{condonComplexityStochasticGames1992} in the $\ell_\infty$-case and a polynomial-time algorithm for ARRIVAL~\cite{haslebacherARRIVALRecursiveFramework2025} in the $\ell_1$-case. 
To be able to implement our algorithm in polynomial time, we would have to be able to find a centerpoint of an intersection of $\ell_p$-halfspaces (the remaining search space) efficiently. However, note that we would not necessarily need a $\frac{1}{d+1}$-centerpoint; a $\frac{1}{\poly(d)}$-centerpoint would suffice.

In Euclidean geometry, such approximate centerpoints can be computed efficiently with a variety of approaches~\cite{cherapanamjeri2024approximatecenterpoints,millerApproximateCenterPoints2009}. One approach for a randomized algorithm that might potentially generalize to our setting is the one based on computing the centerpoint of a representative subset of a point set: given a point set $X\subseteq [0,1]^d$, the set of all halfspaces induces a set system (or \emph{range space}) on $X$. This range space has a Vapnik–Chervonenkis dimension~\cite{vapnik1971vcdimension} of at most $d+1$, as implied by Radon's lemma~\cite{radon1921lemma}. For range spaces with bounded VC-dimension, a random sample of a constant number of points is a so-called \emph{$\epsilon$-approximation}~\cite{mustafa2017handbook} with high probability~\cite{li2000learning,mustafa2017handbook}. Furthermore, the centerpoint of such an $\epsilon$-approximation of the range space of halfspaces is an approximate centerpoint of $X$ itself~\cite{matousek1991centerpoint}. More details about this series of results can be found in Chapter 47 of the \emph{Handbook of Discrete and Computational Geometry}~\cite{mustafa2017handbook}. To generalize this approach to our setting, one would need to bound the VC-dimension of the range space induced by $\ell_p$-halfspaces, and to find an efficient way of sampling a random point from the intersection of $\ell_p$-halfspaces. The first step towards such a sampling procedure may be to study the analogue of Linear Programming feasibility for $\ell_p$-halfspaces: how quickly can we decide whether an intersection of $\ell_p$-halfspaces is empty or not?

\paragraph{Proof of Centerpoint Theorem.}
Our proof of the generalized centerpoint theorem makes use of Brouwer's fixpoint theorem. The classical centerpoint theorem is usually proven using Helly's theorem, which can in turn be proven using Brouwer's fixpoint theorem~\cite{deloera2019mentionofbrouwer}. A proof of the classical centerpoint theorem using Brouwer's fixpoint theorem \emph{directly} in the style of our proof seems to be somewhat folklore, as this strategy is mentioned in Chapter 27 of the \emph{Handbook of Discrete and Computational Geometry}~\cite{zivaljevic2017handbookchapter}, but we were unable to find any reference containing the full proof. 

\subsection{Further Related Work}
\paragraph{Generalized Centerpoint Theorems.} 
The classical centerpoint theorem has been generalized in many different directions. For example, it has been generalized from single centerpoints (generalizing the concept of a median) to multiple centerpoints (generalizing the concept of quantiles)~\cite{pilz2018multiplecenterpoints,ray2007generalizedcenterpoint}. It has also been generalized from centerpoints to centerdisks~\cite{basit2010centerdisks}. The centerpoint theorem has further been generalized to projective spaces~\cite{karasev2014projectivecenterpoints}, which is the only generalization to non-Euclidean geometries that we are aware of.

\paragraph{Fixpoint Theorems.}
Many other fixpoint theorems have been studied from the lens of computational complexity. Finding a fixpoint of a continuous function from a hypercube to itself, as guaranteed by Brouwer's fixpoint theorem, is famously \PPAD-complete~\cite{papadimitriouComplexityParityArgument1994}. Finding a fixpoint of a monotone function from a lattice to itself, as guaranteed by Tarski's fixpoint theorem, is known to lie in \EOPL~\cite{etessamiTarskiTheoremSupermodular2020,goosFurtherCollapsesTFNP2024}, but no hardness result is known. Chang and Lyuu~\cite{chang2010lowerbounds} gave lower bounds on the query complexity of finding fixpoints on \emph{finite} metric spaces, as guaranteed by Banach's and Caristi's fixpoint theorems. The computational problem associated with Caristi's fixpoint theorem on the metric space $([0,1]^3,\ell_\infty)$ has been proven \PLS-complete~\cite{ishizuka2022tfnpFixpoints}. The computational version of Brøndsted's fixpoint theorem on the same metric space is known to lie in \PPAD and to be \CLS-hard~\cite{ishizuka2022tfnpFixpoints}, with the exact complexity still open.

\subsection{Overview and Organization}

We start our exposition in \Cref{sec:preliminaries} by formally defining the problem \continuousProblem\ that we are interested in. We also recall important notions and results from the literature that are needed to understand the rest of the paper. 

We introduce the notion of $\ell_p$-halfspaces in \Cref{ssec:halfspaces} and discuss some of their properties in \Cref{ssec:properties}. Proofs of those properties 
are deferred to \Cref{sec:appendixProofsStructure}. \Cref{ssec:centerpoint_mass_distributions} contains the proof of our generalized centerpoint theorem for mass distributions. We also provide a discrete variant for point sets instead of mass distributions, the proof of which can be found in \Cref{sec:appendixProofCenterpoint}. Finally, we use \Cref{ssec:tightness} to briefly argue that both versions of our generalized centerpoint theorem are tight.

We use \Cref{sec:algorithms} to describe our algorithms. In  \Cref{ssec:continuous_algorithms}, we prove our query upper bound for \continuousProblem\ for all $p \in [1, \infty) \cup \{ \infty\}$. In \Cref{ssec:discrete_algorithms}, we then show how we can ensure that all queries lie on a discrete grid in the $\ell_1$-case, solving what we call the problem \gridToGrid{1}. Finally, we discuss membership in \FPdt of a total search version  \gridToGrid{1} in \Cref{ssec:total_search_problem}. 

\section{Preliminaries}\label{sec:preliminaries}

\begin{definition}[Contraction Map]
    Given a metric space $(X,d_X)$ and a contraction factor $0\leq \lambda < 1$, a function $f : X \rightarrow X$ is called a \emph{$\lambda$-contraction map} (or $\lambda$-contracting) if $d_X(f(x),f(y))\leq \lambda\cdot d_X(x,y)$ holds for all $x,y\in X$. A function is called a \emph{contraction map} if it is $\lambda$-contracting for some $0\leq\lambda<1$.
\end{definition}

\begin{theorem}[Banach Fixpoint Theorem~\cite{banach1922operations}]
    Every contraction map $f : X \rightarrow X$ on a non-empty complete metric space $(X,d_X)$ admits a unique fixpoint $x^\star \in X$, i.e., a unique point satisfying $f(x^\star)=x^\star$.
\end{theorem}

In this paper, we will consider metric spaces of the form $([0,1]^d,\ell_p)$ for some $p \in [1, \infty) \cup \{\infty\}$, where $\ell_p$ denotes the metric induced by the $\ell_p$-norm. 

\begin{definition}[$\ell_p$-Norm]
    The $\ell_p$-norm of $x\in \R^d$ is defined as $||x||_p:=\left(\sum_{i=1}^d |x_i|^p\right)^{1/p}$ for $p \in [1, \infty)$ and $||x||_p:=\max_{i}|x_i|$ for $p = \infty$.
\end{definition}

Concretely, the distance between two points $x, y \in ([0,1]^d,\ell_p)$ is given by $||y-x||_p$. We also use $B^p(x,r) \coloneqq \{y\in\R^d \;\vert\; ||y-x||_p\leq r\}$ to denote the \emph{$\ell_p$-ball} of radius $r$ around $x$.

Finding the exact fixpoint of a contraction map is often infeasible, which is why one usually considers the problem of finding $\epsilon$-approximate fixpoints instead.

\begin{definition}[$\epsilon$-Approximate Fixpoint]
    Given a contraction map $f:X\rightarrow X$ and an $\epsilon>0$, a point $x \in X$ is called an \emph{$\epsilon$-approximate fixpoint} if $d_X(x,f(x)) \leq \epsilon$.
\end{definition}

In the literature, this condition is often called the \emph{residual error criterion}. Note that the contraction property ensures that any $\epsilon$-approximate fixpoint has a distance of at most $\frac{\epsilon}{1-\lambda}$ to the unique fixpoint $x^\star$, thus also bounding the so-called \emph{absolute error} $d_X(x, x^\star)$.

\begin{definition}
    The \continuousProblem problem is to find an $\epsilon$-approximate fixpoint of a $\lambda$-contraction map $f$ on $([0,1]^d,\ell_p)$, given parameters $d \in \N,\epsilon > 0,\lambda \in [0, 1)$ and $p \in [1, \infty) \cup \{ \infty \}$ as well as black-box query access to $f$. The goal is to minimize the number of queries made to $f$.
\end{definition}

We consider an algorithm for \continuousProblem to be \emph{query-efficient} if the number of queries made to $f$ is polynomial in $d$, $\log(\frac{1}{\epsilon})$, and $\log(\frac{1}{1-\lambda})$, and independent of $p$. 

In \continuousProblem, we are allowed to query $f$ at any point $x \in [0, 1]^d$, even at points with irrational coordinates. This might be unsuitable in some applications, and it seems reasonable to also consider a discretized version of this problem, where we are only allowed to make queries to points on a discrete grid. Although we define this discretized version for all $p \in [1, \infty) \cup \{ \infty \}$, we will only really study it for the cases $p \in \{1, \infty\}$ (motivated by applications).

\begin{definition}
    Given an integer $b\geq 1$, the grid $G^d_b \subseteq [0, 1]^d$ is the set of points $x \in [0,1]^d$ with rational coordinates $x_1, \dots, x_d$ of the form $x_i = \frac{k_i}{2^b}$ for integers $k_1, \dots, k_d \in \{0, 1, \dots, 2^b\}$.
\end{definition}

\begin{definition}
    A function $f :G^d_b\rightarrow [0,1]^d$ is called a $\lambda$-contraction grid-map if there exists a $\lambda$-contraction map $f':[0,1]^d \rightarrow [0,1]^d$ such that $f(x)=f'(x)$ for all $x\in G^d_b$ (both $f'$ and $f$ are contracting with respect to the same fixed $\ell_p$-metric). 
\end{definition}
\begin{definition}
    The \gridToGrid{p} problem is to find an $\epsilon$-approximate fixpoint $x \in G^d_b$ of a $\lambda$-contraction grid-map $f:G^d_b \rightarrow [0,1]^d$.
\end{definition}

A priori, \gridToGrid{p} is not guaranteed to have a solution and might therefore not be well-defined. However, it is not hard to see that an $\epsilon$-approximate fixpoint on the grid must exist if the input grid is fine enough.

\begin{lemma}
\label{lemma:lower_bound_b}
    For $b\geq \log_2(\frac{d+d\lambda}{2\epsilon})$, any $\lambda$-contraction grid-map $f:G^d_b\rightarrow [0,1]^d$ admits an $\epsilon$-approximate fixpoint.
\end{lemma}
\begin{proof}
    Let $x^\star$ be the unique fixpoint of some $\lambda$-contraction map $f' : [0, 1]^d \rightarrow [0, 1]^d$ extending $f$. The following calculation shows that any point $x \in [0, 1]^d$ with $||x-x^\star||_p \leq \frac{\epsilon}{1+\lambda}$ must be an $\epsilon$-approximate fixpoint: 
    \begin{align*}
        ||f(x)-x||_p &\leq ||f(x)-f(x^\star)||_p + ||f(x^\star)-x||_p \\
        &\leq \lambda ||x-x^\star||_p+||x-x^\star||_p \\
        &= (1+\lambda)||x-x^\star||_p\leq \epsilon.
    \end{align*}

    We now determine a lower bound on $b$ needed to ensure the existence of such a point $x \in G^d_b$. Observe that since $B^\infty(x^\star, \frac{r}{d}) \subseteq B^p(x^\star, r)$ (follows from $||x||_p \leq d||x||_\infty$) for all $r > 0$ and all $p$, it suffices to prove existence of a grid point in $B^\infty(x^\star, \frac{\epsilon}{d+d\lambda})$. Observe that this is the axis-aligned cube $x^\star + [-\frac{\epsilon}{d+d\lambda},\frac{\epsilon}{d+d\lambda}]^d$. If $2^{-b}\leq \frac{2\epsilon}{d+d\lambda}$, this cube must always contain at least one point of $G^d_b$. Thus, we get the lower bound $b\geq \log_2(\frac{d+d\lambda}{2\epsilon})$.
\end{proof}

\section{\texorpdfstring{$\ell_p$}{lp}-Halfspaces and \texorpdfstring{$\ell_p$}{lp}-Centerpoints}\label{sec:centerpoint_theorems}

In this section, we generalize halfspaces and centerpoints from Euclidean geometry to $\ell_p$-norms for any $p \in [1, \infty) \cup \{ \infty\}$. Note that we sometimes use $\measuredangle(v,w)$ to denote the angle between two vectors $v,w \in \R^d$. We also sometimes use $\overrightarrow{xy}$ for the vector $y - x$. For example, $\measuredangle(\overrightarrow{xy},v)$ denotes the angle between the vectors $y - x$ and $v$.

\subsection{\texorpdfstring{$\ell_p$}{lp}-Halfspaces}\label{ssec:halfspaces}

There are many different ways of defining a halfspace in Euclidean geometry ($p = 2$). One natural way that will appear in our algorithms for finding the fixpoints of contraction maps (see \Cref{sec:algorithms}) is to define a halfspace using two distinct points  $x, y \in \R^d$: a point $z \in \R^d$ is considered to be in the halfspace if and only if $\norm{x - z}_2 \leq \norm{y - z}_2$. This directly generalizes to arbitrary $p \in [1, \infty) \cup \{\infty\}$ as follows.

\begin{definition}[Bisector $\ell_p$-Halfspace]
    For fixed $p \in [1, \infty) \cup \{\infty\}$ and distinct points $x, y \in \R^d$, the bisector $\ell_p$-halfspace $\bisecH^p_{x,y} \subseteq \R^d$ is defined as
    \[
        \bisecH^p_{x, y} \coloneqq \{ z \in \R^d \mid ||x  - z||_p \leq ||y - z||_p \}.
    \]
\end{definition}

In Euclidean geometry, rather than defining a halfspace using the principle of bisection, we can alternatively define it by a point $x$ on the boundary and a normal vector $v \in S^{d - 1}$, making use of the Euclidean inner product: a point $z \in \R^d$ belongs to the halfspace if and only if $\langle v, z-x \rangle \geq 0$. In fact, the $\ell_2$-norm is special because it is induced by an inner product. In general, $\ell_p$-norms are not induced by an inner product, and therefore this definition does not directly generalize. However, we can avoid the inner product with the following observation: a point $z$ satisfies $\langle v,z-x\rangle\geq 0$ if and only if $\norm{x - z}_2 \leq \norm{x - \epsilon v - z}_2$ holds for all $\epsilon > 0$. In a way, this definition can be seen as the limit of the bisector definition with $y=x-\epsilon v$ and $\epsilon\rightarrow 0$. Stated like this, we can generalize the definition to arbitrary $p \in [1, \infty) \cup \{\infty\}$.

\begin{restatable}[Limit $\ell_p$-Halfspace]{definition}{limithalfspace}\RestateRemark
    For fixed $p \in [1, \infty) \cup \{\infty\}$, point $x \in \R^d$, and direction $v \in \R^d$ with $v \neq 0$, the $\ell_p$-halfspace $\limitH^p_{x,v} \subseteq \R^d$ through $x$ in the direction of $v$ is defined as
    \[
        \limitH^p_{x, v} \coloneqq \{ z \in \R^d \mid \forall \epsilon > 0 \, : \, ||x - z||_p \leq ||x - \epsilon v - z||_p \}.
    \]
\end{restatable}

Observe that scaling the direction $v$ with a positive scalar does not change the halfspace $\limitH^p_{x, v}$. Hence, we usually assume $v \in S^{d - 1} \subseteq \R^d$. In fact, we will frequently use the following characterization for containment in a limit $\ell_p$-halfspace.
\begin{restatable}{observation}{obscharacterizationcontainment}
\label{obs:characterization_containment}\RestateRemark
    For a given limit $\ell_p$-halfspace $\limitH^p_{x, v}$ and point $z \in \R^d$, let $R_-$ be the open ray from $x$ in direction $-v$, and let $B_z$ be the smallest closed $\ell_p$-ball with center $z$ that contains $x$. Then we have
    \[
        z \in \limitH^p_{x, v} \iff R_- \cap B_z^\circ = \emptyset,
    \]
    where $B_z^\circ$ denotes the interior of $B_z$.
\end{restatable}

For our centerpoint theorem, we will exclusively work with limit halfspaces. However, as mentioned before, bisector halfspaces naturally appear in the analysis of our algorithms in \Cref{sec:algorithms}. Thus, we will need the following observation that allows us to translate between the two.

\begin{observation}\label{obs:bisectorSubsetLimit}
    Let $p \in [1, \infty) \cup \{\infty\}$, $x \in \R^n$, and $v \in S^{d -1}$. The limit $\ell_p$-halfspace $\limitH^p_{x,v}$ is the intersection of all bisector $\ell_p$-halfspaces $\bisecH^p_{x,x-\epsilon v}$ for $\epsilon>0$. Thus, we have $\limitH^p_{x,v}\subseteq \bisecH^p_{x,x-\epsilon v}$ for all $\epsilon > 0$.
\end{observation}

For the rest of this section, we will exclusively work with limit $\ell_p$-halfspaces and refer to them simply as $\ell_p$-halfspaces. In particular, we will first establish some properties of $\ell_p$-halfspaces that we will then use to prove our centerpoint theorem.

\subsection{Properties of \texorpdfstring{$\ell_p$}{lp}-Halfspaces}\label{ssec:properties}

We now collect some useful properties of $\ell_p$-halfspaces. The proofs of all the following properties and additional insights into $\ell_p$-halfspaces can be found in \Cref{sec:appendixProofsStructure}. 

We start with a simple observation for the case when the direction $v$ is a standard unit vector (i.e., $v$ is parallel to one of the coordinate axes). It turns out that for those directions, $\ell_p$-halfspaces are no different from the classical $\ell_2$-halfspaces.

\begin{restatable}{lemma}{lemmaaxisalignedhalfspace}\label{lemma:axis_aligned_halfspace}\RestateRemark
    For any $p\in [1,\infty)$, any $x \in \R^d$, and any direction $v \in S^{d - 1}$ parallel to one of the coordinate axes, we have $\limitH^p_{x, v} = \limitH^2_{x, v}$. 
\end{restatable}
Note that for $\ell_\infty$, \Cref{lemma:axis_aligned_halfspace} does not hold, and we only have $\limitH^\infty_{x, v} \supset \limitH^2_{x, v}$. 

\Cref{lemma:axis_aligned_halfspace} really only describes a handful of special cases. For most directions $v$, $\ell_p$-halfspaces are very different to classical $\ell_2$-halfspaces. In particular, they are in general not convex.
Still, we can give a somewhat nice qualitative description of their shape.

\begin{restatable}{lemma}{lemstructureoflimithalfspace}\label{lem:structureoflimithalfspace}\RestateRemark
    For any $p\in [1,\infty) \cup \{\infty\}$, any $\ell_p$-halfspace $\limitH^p_{x,v}$ is a union of rays originating in $x$. For $p\in (1,\infty)$, the boundary of any $\ell_p$-halfspace $\limitH^p_{x,v}$ is a union of lines through $x$.
\end{restatable}

Note that we will not really use the observation about the boundary of $\ell_p$-halfspaces for $p \in (1, \infty)$, but find it an interesting addition to the lemma. Unfortunately, it breaks down in some degenerate cases if $p \in \{1, \infty\}$ (see, e.g.,~\Cref{fig:degenerate}).

\begin{figure}[ht]
    \centering
    \includegraphics[]{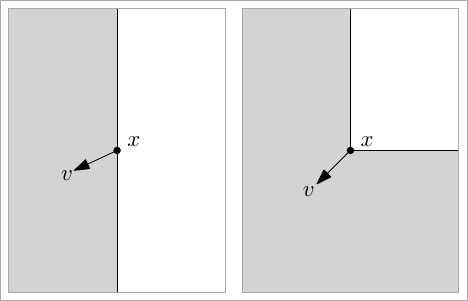}
    \caption{Examples of \Cref{lem:structureoflimithalfspace}. Two $\ell_1$-halfspaces, both are unions of rays starting at $x$, but only the left one (with non-degenerate $v$) has a boundary consisting of a union of lines through $x$. }
        \label{fig:degenerate}
\end{figure}

\Cref{lem:structureoflimithalfspace} provides a qualitative description of $\ell_p$-halfspaces, but does not really give us any concrete points that must be contained in or outside of a given $\ell_p$-halfspace. For example, it seems intuitive that points lying roughly in the direction of $v$ from $x$ should be contained in $\limitH^p_{x, v}$, while points lying roughly in the direction of $-v$ from $x$ should never be contained. The following lemma formalizes this (see also \Cref{fig:cones}).

\begin{restatable}{lemma}{lemanglesonlimithalfspaces}\label{lem:anglesonlimithalfspaces}\RestateRemark
    Let $p\in [1,\infty) \cup \{\infty\}$ and $\limitH^p_{x, v}$ be arbitrary. All $z \in \limitH^p_{x, v}$ satisfy $\measuredangle(\overrightarrow{xz},v) \leq \pi-\sqrt{\nicefrac{1}{d}}$. Similarly, all $z \in \R^d$ with $\measuredangle(\overrightarrow{xz},v) \leq \sqrt{\nicefrac{1}{d}}$ must be contained in $\limitH^p_{x, v}$.
\end{restatable}

\begin{cfigure}[ht] 
    \centering
    \includegraphics[]{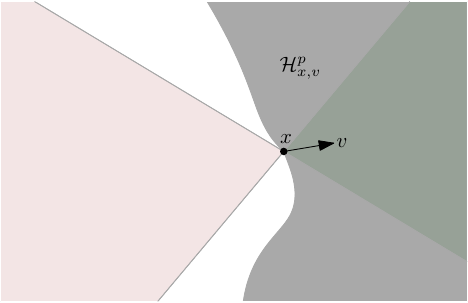}
    \caption{Sketch of \Cref{lem:anglesonlimithalfspaces}. No point in the red cone is contained in the $\ell_p$-halfspace $\limitH^{p}_{x,v}$, but all points in the green cone are. }
    \label{fig:cones}
\end{cfigure}

With \Cref{lem:anglesonlimithalfspaces}, it now also seems clear that any $\ell_p$-halfspace $\limitH^p_{x, -x}$ with $x$ far away from the origin should contain every small compact set $C$ around the origin. \Cref{lemma:pull_towards_zero} makes this precise for $C = [0, 1]^d$. 

\begin{restatable}{corollary}{lemmapulltowardszero}
\label{lemma:pull_towards_zero}\RestateRemark
    For arbitrary $p \in [1, \infty) \cup \{\infty\}$ and $x \in \R^d$ with $\norm{x}_2 > 2d$, we have $[0,1]^d \subseteq \limitH^p_{x, -x}$.
\end{restatable}

\begin{figure}[ht]
    \centering
    \includegraphics[]{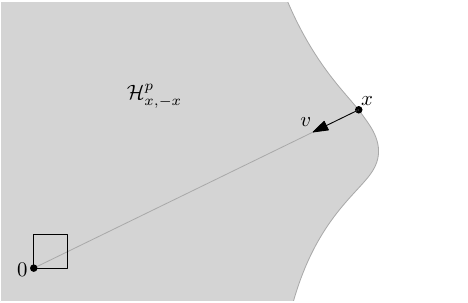}
    \caption{Sketch of \Cref{lemma:pull_towards_zero}. The cube $[0,1]^d$ is contained in all $\limitH^p_{x,-x}$ for $||x||_2$ large enough.}
\end{figure}

Finally, we will need some results on the interaction of $\ell_p$-halfspaces with mass distributions, defined as follows.

\begin{definition}[Mass Distribution]
    We call a measure $\mu$ on $\R^d$ a mass distribution if it is absolutely continuous with respect to the Lebesgue measure on $\R^d$ and satisfies $\mu(\R^d) < \infty$.
\end{definition}

Note that absolute continuity of $\mu$ says that all measurable sets $A$ with Lebesgue measure $0$ must also satisfy $\mu(A) = 0$. By the Radon–Nikodym theorem, any such mass distribution must have a density, and so one can think of mass distributions as probability distributions with a density (by assuming $\mu(\R^d) = 1$, without loss of generality). Concretely, we will adopt this point of view in the proofs of \Cref{lemma:continuity_halfspace_mass} and \Cref{lemma:pulling_directions_positive_measure} (see \Cref{ssec:properties_II}). 

We can now study how $\mu(\limitH^p_{x, v})$ behaves with respect to its arguments $x$ and $v$. For example, the next lemma says that the mass $\mu(\limitH^p_{x, v})$ of an $\ell_p$-halfspace is continuous under translation of the defining point~$x$.

\begin{restatable}{lemma}{lemmacontinuityhalfspacemass}
\label{lemma:continuity_halfspace_mass}\RestateRemark
    The function $f(x) \coloneqq \mu(\limitH^p_{x, v})$ is continuous in $x \in \R^d$ for all mass distributions $\mu$, all $v \in S^{d - 1}$, and all $p \in [1, \infty) \cup \{\infty\}$.
\end{restatable}

For $p \in (1, \infty)$, one could also show continuity with respect to $v \in S^{d - 1}$ rather than $x$. Unfortunately, this breaks down for $p \in \{1, \infty\}$ again in some degenerate cases. However, it turns out that a weaker property will suffice for our needs. Concretely, for fixed $x \in \R^d$ and some threshold $t > 0$, consider the set of directions $V_x \coloneqq \{ v \in S^{d - 1} \mid \mu(\limitH^p_{x, v}) < t\}$ corresponding to halfspaces $\limitH^p_{x, v}$ with strictly less than $t$ mass. The following lemma says that $V_x$ is an open subset of $S^{d - 1}$.

\begin{restatable}{lemma}{lemmapullingdirectionspositivemeasure}
\label{lemma:pulling_directions_positive_measure}\RestateRemark
    The set $V_x = \{ v \in S^{d - 1} \mid \mu(\limitH^p_{x, v}) < t\}$ is an open subset of $S^{d - 1}$ for all mass distributions $\mu$, all $t > 0$, all $p \in [1, \infty) \cup \{\infty\}$, and all $x \in \R^d$.
\end{restatable}
    
\subsection{\texorpdfstring{$\ell_p$}{lp}-Centerpoints of Mass Distributions}\label{ssec:centerpoint_mass_distributions}

Using our notion of limit halfspaces, we can state the classical (Euclidean) centerpoint theorem for mass distributions as follows.

\begin{theorem}[Euclidean Centerpoint Theorem for Mass Distributions~\cite{radoTheoremGeneralMeasure1946}] 
\label{theorem:classical_centerpoint}
    Let $\mu$ be a mass distribution on $\R^d$ with bounded support. There exists $c \in \R^d$ such that $\mu(\limitH^2_{c, v}) \geq \frac{1}{d + 1}\mu(\R^d)$ for all $v \in S^{d - 1}$.
\end{theorem}

We call such a point $c$ an \emph{$\ell_2$-centerpoint} of $\mu$. In this section, we will prove existence of an \emph{$\ell_p$-centerpoint} of $\mu$ for all $p \in [1, \infty) \cup \{ \infty \}$, as stated in the following theorem.

\begin{theorem}[$\ell_p$-Centerpoint Theorem for Mass Distributions]
\label{theorem:p-centerpoint}
    Let $\mu$ be a mass distribution on $\R^d$ with bounded support, and let $p \in [1, \infty) \cup \{\infty\}$ be arbitrary. There exists $c\in\R^d$ such that $\mu(\limitH^p_{c, v}) \geq \frac{1}{d + 1}\mu(\R^d)$ for all $v \in S^{d - 1}$.
\end{theorem}

A classical and simple proof of \Cref{theorem:classical_centerpoint} uses Helly's theorem~\cite{Helly1923} to show that all $\ell_2$-halfspaces containing strictly more than a $\frac{d}{d+1}$-fraction of $\mu(\R^d)$ must have a non-empty intersection. Choosing $c$ in this intersection guarantees the required property. However, this proof breaks down if we consider $\ell_p$-halfspaces: they are in general non-convex and thus Helly's theorem is no longer applicable. Instead, we prove \Cref{theorem:p-centerpoint} with Brouwer's fixpoint theorem. The existence of a proof of \Cref{theorem:classical_centerpoint} via Brouwer's fixpoint theorem seems to be folklore\footnote{Such a proof is vaguely outlined in~\cite{zivaljevic2017handbookchapter}.}, but we were unable to find a reference containing the full proof. We briefly recall Brouwer's fixpoint theorem.

\begin{theorem}[Brouwer's fixpoint theorem~\cite{brouwerUeberAbbildungMannigfaltigkeiten1911}]
Let $C\subseteq \R^d$ be a non-empty compact convex set, and let $f:C\rightarrow C$ be a continuous function. Then $f$ has a fixpoint.
\end{theorem}

We will now use Brouwer's fixpoint theorem to prove \Cref{theorem:p-centerpoint}. For this, we first define a function $F:\R^d\rightarrow \R^d$ such that any fixpoint of $F$ must be an $\ell_p$-centerpoint of $\mu$. Concretely, the intuition behind this function is that any point $x$ is ``pushed'' in all directions $v$ for which $\limitH^p_{x,-v}$ does not contain a sufficient fraction of mass simultaneously. We achieve this by defining $F_i(x)$ for all $i\in [d]$ using an integral over all directions, i.e.,
\[
        F_i(x) \coloneqq x_i + \int_{S^{d - 1}} v_i \max \left( \frac{1}{d + 1}\mu(\R^d) - \mu(\limitH^p_{x, -v}), 0 \right) dv.
\]

Since we eventually want to apply Brouwer's fixpoint theorem to $F$, let us check that $F$ is continuous.

\begin{lemma}
\label{lemma:continuity_of_f}
    Let $p \in [1, \infty) \cup \{\infty\}$ be arbitrary. Let $\mu$ be a mass distribution on $\R^d$ with bounded support. Then the function $F : \R^d \rightarrow \R^d$ defined above is continuous. 
\end{lemma}
\begin{proof}
    Fix $x \in \R^d$ and let $(x_n)_{n \in \N}$ be an arbitrary sequence converging to $x$. By \Cref{lemma:continuity_halfspace_mass}, the functions
    \[
        g_n(v) \coloneqq v_i \max \left( \frac{1}{d + 1}\mu(\R^d) - \mu(\limitH^p_{x_n, -v}), 0 \right) 
    \]
    converge pointwise to $g(v) = v_i \max \left( \frac{1}{d + 1}\mu(\R^d) - \mu(\limitH^p_{x, -v}), 0 \right)$ for all $i \in [d]$. Since they are also globally bounded, we can apply the dominated convergence theorem to exchange limit and integral, yielding continuity of $F$.
\end{proof}

We next turn our attention to the fixpoints of $F$. Concretely, we wish to show that any fixpoint $x$ of $F$ must be an $\ell_p$-centerpoint of $\mu$. To argue this, we analyze the set 
\[
    V_x \coloneqq \{v \in S^{d - 1} \mid \max \left( \frac{1}{d + 1}\mu(\R^d) - \mu(\limitH^p_{x, -v}), 0 \right) > 0\}
\]
of directions which contribute towards the integral. We will show that it is impossible for $\conv(V_x)$ to contain $0$ (for any $x \in \R^d$). Later, we will then use this to argue that the integral in the definition of $F$ cannot be $0$ without $V_x$ being empty. This in turn allows us to conclude that fixpoints of $F$ are $\ell_p$-centerpoints of $\mu$. 

To prove that $0\not\in\conv(V_x)$, we need the following lemma.

\begin{lemma}
\label{lemma:convex_combinations} 
    Assume that $0 \in \R^d$ is contained in the convex hull of $k \leq d + 1$ points $v_1, \dots, v_k \in S^{d - 1}$. Let $z \in \R^d$ and $p \in [1, \infty) \cup \{\infty\}$ be arbitrary. Then there exists $j \in [k]$ such that $z \in \limitH^p_{0, -v_j}$.
\end{lemma}
\begin{proof}
    Let $R_1,\ldots,R_k$ be the open rays from $0$ in the directions $v_1,\ldots,v_k$, respectively. Let us now grow an $\ell_p$-ball $B_z$ centered at $z$, exactly until the ball contains $0$. Recall that by \Cref{obs:characterization_containment}, we have $z\in \limitH^p_{0,-v_j}$ if and only if $B^\circ_z$ does not intersect $R_j$. Now observe that if we had $B^\circ_z \cap R_j \neq \emptyset$ for all $j \in [k]$, then this would imply $0 \in B^\circ_z$ by convexity of $B^\circ_z$. But then, $B_z$ cannot have been the minimum-radius ball centered at $z$ containing $0$, a contradiction.
\end{proof}

Let us now see how this lemma implies $0\not\in\conv(V_x)$ for all $x \in \R^d$.

\begin{corollary}\label{corollary:convexhull}
    Let $p \in [1, \infty) \cup \{\infty\}$ be arbitrary and let $\mu$ be a mass distribution on $\R^d$. For any point $x \in \R^d$, the set $V_x=\{v \in S^{d - 1} \mid \max \left( \frac{1}{d + 1}\mu(\R^d) - \mu(\limitH^p_{x, -v}), 0 \right) > 0\}$ does not contain $0$ in its convex hull.
\end{corollary}
\begin{proof}
If $0$ would lie in $\conv(V_x)$, this would mean that $V_x$ is non-empty. In particular, using Carathéodory's theorem, we would get a set of $k \leq d + 1$ vectors $v_1, \dots, v_{k} \in V_x$ with $0 \in \conv(v_1, \dots, v_{k})$. Since these vectors are in $V_x$, we get on the one hand that
    \[
        \sum_{i = 1}^k \mu(\limitH^p_{x, -v_i}) < \frac{k}{d + 1}\mu(\R^d) \leq \mu(\R^d).
    \]
On the other hand, \Cref{lemma:convex_combinations} guarantees that every $z \in \R^d$ is contained in at least one of the halfspaces $\limitH^p_{x, -v_1}, \dots, \limitH^p_{x, -v_k}$. Thus, we can derive 
    \[
        \sum_{i = 1}^k \mu(\limitH^p_{x, -v_i}) \geq \mu \left( \bigcup_{i = 1}^k \limitH^p_{x, -v_i} \right) = \mu(\R^d),
    \]
    a contradiction.    
\end{proof}

We are now ready to put all of this together to get a proof of \Cref{theorem:p-centerpoint}.
The main technical thing that we still have to do, is to restrict $F$ to a compact convex set (so that we can apply Brouwer's fixpoint theorem). 

\begin{proof}[Proof of Theorem~\ref{theorem:p-centerpoint}]
    Without loss of generality, assume that the bounded support of $\mu$ is contained in the box $[0,1]^d$. To apply Brouwer's fixpoint theorem, we need a function going from a compact convex set $C$ to itself, rather than from $\R^d$ to $\R^d$ like our function $F$. We create such a function $F_C$ by defining a large compact convex set $C\supseteq [0,1]^d$, and then restricting the function $F$ to $C$ using projection. More formally, we use the restricted function $F_C:C\rightarrow C$ with $F_C(x)$ defined as the projection of $F(x)$ onto $C$ in the direction of the origin (for any $x\in C$). As our set $C$, we choose the Euclidean ball of radius $1+2d$ around the origin, since this allows us to apply \Cref{lemma:pull_towards_zero} for all $x$ on the boundary of $C$.
    
    By continuity of $F$, it follows that $F_C$ is also continuous. Thus, we can use Brouwer's fixpoint theorem, which tells us that there exists $c \in C$ with $F_C(c) = c$. We now prove that $c$ is a centerpoint by distinguishing the two cases $F(c)=F_C(c)$ and $F(c)\neq F_C(c)$.
    
    In the first case, we must have that
    \[
        \int_{S^{d - 1}} v_i \max \left( \frac{1}{d + 1}\mu(\R^d) - \mu(\limitH^p_{c, - v}), 0 \right) dv = 0
    \]
    for all $i \in [d]$.
    Now consider the set of directions $V_c = \{v \in S^{d - 1} \mid \max \big( \frac{1}{d + 1}\mu(\R^d) - \mu(\limitH^p_{c, - v}), 0 \big) > 0\}$. By \Cref{lemma:pulling_directions_positive_measure}, this is an open subset of $S^{d - 1}$. 
    Moreover, $0$ is not contained in $\conv(V_c)$ by \Cref{corollary:convexhull}. However, this implies that $V_c$ must be empty: if $V_c$ were non-empty, we could find a separating hyperplane between the two convex sets $\{0 \}$ and $\conv(V_c)$. But then, the definition of $V_c$ and its openness contradicts that the above integral is zero (there would be a non-zero net push in the direction of the normal vector defining the separating hyperplane). 
    Therefore, since $V_c$ must be empty, the expression $\max( \frac{1}{d + 1}\mu(\R^d) - \mu(\limitH^p_{c, -v}), 0 )$ must be $0$ for all $v \in S^{d - 1}$, and $c$ is a centerpoint. 

    In the second case, $F(c)\neq F_C(c)$, we know that $F_C$ ``used'' the projection to map $c$. We thus know that $c=F_C(c)$ lies on the boundary of $C$. Furthermore, since the projection onto $C$ is towards the origin, we get $F(c)=(1+\epsilon)c$ for some $\epsilon>0$. Concretely, we must have
    \[
        \int_{S^{d - 1}} v_i \max \left( \frac{1}{d + 1}\mu(\R^d) - \mu(\limitH^p_{c, - v}), 0 \right) dv = \epsilon c_i
    \]
    for all $i \in [d]$.
    Considering the set $V_c$ again, we can see that there hence must exist some $\delta > 0$ with $\delta c \in\conv (V_c)$: otherwise, we could find a hyperplane separating $\conv (V_c)$ from the line segment $[0, \frac{c}{\norm{c}_2}]$, yielding a contradiction. 
    But since $C$ was chosen large enough, we also get $-\frac{c}{||c||_2}\in V_c$ by \Cref{lemma:pull_towards_zero}.
    Thus, we see that $0\in\conv (V_c)$. This contradicts \Cref{corollary:convexhull}, and hence this case cannot occur. We conclude that every fixpoint of $F_C$ is a fixpoint of $F$, and thus every fixpoint of $F_C$ must be a centerpoint.
\end{proof}

\Cref{theorem:classical_centerpoint} can be adapted to point sets instead of mass distributions (see \Cref{theorem:discrete_p-centerpoint} below). The proof is not difficult but a bit technical, which is why we defer it to \Cref{sec:appendixProofCenterpoint}. The main idea is to put a ball of small radius around each point and to apply \Cref{theorem:p-centerpoint}. Letting the radius go to zero, we obtain a sequence of centerpoints. A subsequence of this sequence must converge to a discrete $\ell_p$-centerpoint of the point set.

\begin{restatable}[$\ell_p$-Centerpoint Theorem for Finite Point Sets]{theorem}{theoremdiscretepcenterpoint}
\label{theorem:discrete_p-centerpoint}\RestateRemark
    Let $p \in [1, \infty) \cup \{\infty\}$ be arbitrary, and let $P \subseteq \R^d$ be a finite set of points. There exists a point $c$ such that $|\limitH^p_{c, v} \cap P| \geq \frac{|P|}{d + 1}$ for all $v \in S^{d - 1}$. 
\end{restatable}

Finally, we observe that unless $p=\infty$, the centerpoint in both the discrete and continuous setting must lie inside any axis-aligned bounding box of the point set $P$ or the support of $\mu$, respectively.

\begin{lemma}
\label{obs:centerpoint_in_bounding_box}
    Let $p\in [1,\infty)$. Let $B$ be an axis-aligned bounding box that contains the support of $\mu$ (in the case of \Cref{theorem:p-centerpoint}) or all of $P$ (in the case of \Cref{theorem:discrete_p-centerpoint}), respectively. Any $\ell_p$-centerpoint guaranteed by either theorem must lie inside $B$.
\end{lemma}
\begin{proof}
    Towards a contradiction, let $c$ be an $\ell_p$-centerpoint of $\mu$ (or $P$, respectively) not contained in $B$. Assume without loss of generality that every point $x$ in the support of $\mu$ (or in $P$) satisfies $x_1 < c_1$. Then, by \Cref{lemma:axis_aligned_halfspace}, we have that $\mu(\limitH^p_{c, e_1})=0$ (or $\limitH^p_{c, e_1} \cap P = 0$), where $e_1$ is the first standard unit vector. This contradicts the centerpoint property of $c$.
\end{proof}

\subsection{Tightness of Centerpoint Theorems}\label{ssec:tightness}

We want to remark that the fraction $\frac{1}{d + 1}$ in \Cref{theorem:p-centerpoint} and \Cref{theorem:discrete_p-centerpoint} is tight for all $p\in (1,\infty)$. To prove this, we use a construction that has also been used to prove tightness of the classical Euclidean centerpoint theorem. We will restrict ourselves to the discrete setting of \Cref{theorem:discrete_p-centerpoint}, but tightness of \Cref{theorem:p-centerpoint} follows as well because any better bound for mass distributions could be used to get a better bound for point sets by following the proof of \Cref{theorem:discrete_p-centerpoint} (see \Cref{sec:appendixProofCenterpoint}).

\begin{lemma}
    For every $d \in \N$ and every $p\in (1,\infty)$, there exists a point set $P_d \subseteq [0, 1]^d$, such that every point $c \in \R^d$ has a direction $v \in S^{d - 1}$ with $|\limitH^p_{c,v} \cap P_d|\leq \frac{|P_d|}{d+1}$.
\end{lemma}
\begin{proof}
    For $p\in (1,\infty)$, let $P_d=\{ 0,e_1,\ldots,e_d\} \subseteq [0, 1]^d$, where $e_i$ is the $i$-th standard unit vector. Let $c\neq 0$ be arbitrary. Then $c$ cannot lie on all axis-aligned facets of the convex hull of $P_d$ simultaneously. Thus, by \Cref{lemma:axis_aligned_halfspace} there must exist a direction $v\in\{-e_1,\ldots,-e_d,e_1,\ldots,e_d\}$ such that $\limitH^p_{c,v}$ does not contain any vertices of a facet that $c$ does not lie on, and therefore $|\limitH^p_{c,v}\cap P_d|\leq 1 = \frac{|P_d|}{d+1}$.
    If we instead have $c=0$, then one can check by calculation that $\limitH^p_{c,v}\cap P_d = \{0\}$ for $v=(-\sqrt{\nicefrac{1}{d}},\ldots,-\sqrt{\nicefrac{1}{d}}) \in S^{d - 1}$.
\end{proof}

Unfortunately, this construction does not work for $p \in \{1, \infty\}$. In fact, we already know from~\cite{chenComputingFixedPoint2024} that it cannot work for $p = \infty$, since they provide better centerpoints in that case. However, tightness in the case of $p = 1$ remains open.

\section{Finding Fixpoints of \texorpdfstring{$\ell_p$}{lp}-Contraction Maps}\label{sec:algorithms}

In this section, we describe our algorithms for the continuous problem \continuousProblem and explain how, in the $\ell_1$-case, the algorithm can be adapted to the discretized setting \gridToGrid{1}. Note that an algorithm for \continuousProblem[\infty] as well as a suitable rounding strategy, implying an algorithm for \gridToGrid{\infty}, has already been provided by Chen, Li, and Yannakakis~\cite{chenComputingFixedPoint2024}. We will therefore restrict our attention to $p\in[1,\infty)$.

\subsection{Solving \texorpdfstring{$\ell_p$-ContractionFixpoint}{p-ContractionFixpoint}}\label{ssec:continuous_algorithms}

Our algorithm works as follows: we maintain a search space $M$ that is guaranteed  to always contain the fixpoint $x^\star$. At the beginning of the algorithm, we simply set $M=[0,1]^d$. We then iteratively query the centerpoint $c$ of our remaining search space $M$. Concretely, we use the measure $\vol$ (which we simply call \emph{volume}) defined by the Lebesgue measure (such that $\vol([0,1]^d)=1$), with its support restricted to $M$.
Whenever we query the centerpoint $c$ of $M$, we get to discard at least a $\frac{1}{d + 1}$-fraction of the search space, because $x^\star$ must lie closer to $f(c)$ than to $c$ itself. With each query, the volume of $M$ is thus multiplied with a factor of at most $\frac{d}{d + 1}$. We terminate once we happen to query an $\epsilon$-approximate fixpoint. 
It remains to prove that this must happen before the search space $M$ gets too small. To that end, we show in the next lemma that, whenever we query a point that is not an $\epsilon$-approximate fixpoint, a ball of some radius $r_{\epsilon, \lambda}$ around $x^\star$ cannot be discarded, and thus has to remain in the search space $M$. 

Note that while the centerpoint theorems use limit $\ell_p$-halfspaces, we will now use bisector halfspaces in the analysis of our algorithm. \Cref{obs:bisectorSubsetLimit} allows us to translate between the two.

\begin{lemma}\label{lem:ballaroundfixpointstays}
    Let $p \in [1, \infty)$ be arbitrary. Let $x^\star$ be the unique fixpoint of the $\lambda$-contracting map $f : [0,1]^d \rightarrow [0,1]^d$. Let $x \in [0, 1]^d$ be arbitrary. If $x$ is not an $\epsilon$-approximate fixpoint of $f$, then 
    \[
        B^p(x^\star, r_{\epsilon, \lambda}) \cap \bisecH^p_{x,f(x)}=\emptyset
    \]
    for $r_{\epsilon, \lambda} = \frac{\epsilon - \epsilon \lambda}{2+2\lambda}$.
\end{lemma}
\begin{proof}   
    Let $z \in B^p(x^\star, r_{\epsilon, \lambda})$ be arbitrary. We need to show $||z-x||_p > ||z -f(x)||_p$. Let us first collect some facts. Concretely, we know that $x$ is not an $\epsilon$-approximate fixpoint, $z\in B^p(x^\star,r_{\epsilon,\lambda})$, and that $f$ is $\lambda$-contracting. This gives us the following inequalities:
    \begin{align}
        ||f(x) - x||_p &> \epsilon \label{ineq:noepsapproximate}\\
        ||z - x^\star||_p &\leq r_{\epsilon, \lambda} \\
        ||x^\star-f(x)||_p &\leq \lambda ||x^\star-x||_p.\label{ineq:contractivefunction}
    \end{align}
    Combining these using the triangle inequality, we also get
    \begin{align}
        ||z - x||_p &\geq ||x^\star - x||_p - r_{\epsilon, \lambda} \label{ineq:fprimeinball}\\
        ||z - f(x)||_p &\leq ||x^\star-f(x)||_p + r_{\epsilon, \lambda}.\label{ineq:45}
    \end{align}
    Before putting everything together, we also want to lower bound $||z - x||_p$. We first lower bound $||x^\star - x||_p$ using the calculation
    \begin{equation}
        \epsilon \overset{(\ref{ineq:noepsapproximate})}{<} ||x - f(x)||_p \leq ||x^\star-x||_p + ||x^\star - f(x)||_p \overset{(\ref{ineq:contractivefunction})}{\leq} (1+\lambda) ||x^\star-x||_p. \label{ineq:fxlowerbound}
    \end{equation}
    Combining \Cref{ineq:fprimeinball,ineq:fxlowerbound}, we get
    \begin{equation}
        ||z-x||_p > \frac{\epsilon}{1+\lambda} - r_{\epsilon, \lambda}.\label{ineq:zxbound}
    \end{equation}
    We can now put everything together to obtain
    \begin{align*}
    ||z-f(x)||_p \overset{(\ref{ineq:45})}&{\leq} ||x^\star-f(x)||_p + r_{\epsilon, \lambda} \\
        \overset{(\ref{ineq:contractivefunction})}&{\leq} \lambda ||x^\star-x||_p + r_{\epsilon, \lambda}  \\
        \overset{(\ref{ineq:fprimeinball})}&{\leq} \lambda ||z-x||_p + (1+\lambda) r_{\epsilon, \lambda} \\
    &= ||z - x||_p - (1-\lambda) ||z-x||_p + (1+\lambda) r_{\epsilon, \lambda}\\
    \overset{(\ref{ineq:zxbound})}&{<} ||z-x||_p - (1-\lambda) \left(\frac{\epsilon}{1+\lambda}-r_{\epsilon, \lambda}\right) + (1+\lambda) r_{\epsilon, \lambda} \\
    &= ||z-x||_p + 2r_{\epsilon, \lambda} - \frac{\epsilon - \epsilon \lambda}{1+\lambda} \\
    &= ||z-x||_p,
    \end{align*}
    and we conclude $||z-x||_p>||z-f(x)||_p$ and thus $z \notin \bisecH^p_{x,f(x)}$.
\end{proof}

We now have all the ingredients to prove \Cref{thm:mainContinuous}.

\begin{theorem}\label{thm:mainContinuous}
    For every $p \in [1, \infty) \cup \{ \infty\}$, an $\epsilon$-approximate fixpoint of a $\lambda$-contracting (in $\ell_p$-norm) function $f:[0,1]^d\rightarrow [0,1]^d$ can be found using $\bigO(d^2 ( \log\frac{1}{\epsilon} + \log\frac{1}{1-\lambda}))$ queries.
\end{theorem}
\begin{proof}
    Recall that the case $p=\infty$ has been solved by Chen, Li, and Yannakakis~\cite{chenComputingFixedPoint2024}.

    We begin with the search space $M=[0,1]^d$ with $\vol(M)=1$. We then repeatedly query the centerpoint $c$ of $M$ (or rather the measure induced by $M$), guaranteed to exist by \Cref{theorem:p-centerpoint}, and guaranteed to lie inside $[0,1]^d$ by \Cref{obs:centerpoint_in_bounding_box}. We terminate if $c$ is an $\epsilon$-approximate fixpoint. Otherwise, we remove the bisector halfspace $H_{c,f(c)}$ from $M$. 
    For each non-terminating query, we get the guarantee that the volume of $M$ decreases to at most a $\frac{d}{d+1}$-fraction of its previous volume, where we use \Cref{obs:bisectorSubsetLimit} to translate the centerpoint guarantee of \Cref{theorem:p-centerpoint} from limit halfspaces to bisector halfspaces. 
    
    Now recall that \Cref{lem:ballaroundfixpointstays} guarantees that $B^p(x^*,r_{\epsilon,\lambda})$ stays in $M$ as long as we do not query an $\epsilon$-approximate fixpoint. Thus, we know that we must terminate before having reached $\vol(M) < \vol(B^p(0,r_{\epsilon, \lambda}))$. Therefore, to compute our final query bound, we first note that it is well-known that the volume of the $\ell_p$-norm ball $B^p(0,r_{\epsilon, \lambda})$ can be bounded from below by 
    \[
        \vol(B^p_{0,r_{\epsilon, \lambda}}) \geq \frac{2^d}{d!}r_{\epsilon, \lambda}^d = \frac{2^d}{d!}(\frac{\epsilon - \epsilon\lambda}{2+2\lambda})^d.
    \]
    Finally, we can upper bound the number $k$ of queries by
    \[
        k\leq \log_{\frac{d}{d+1}}\left(\frac{2^d}{d!}\left(\frac{\epsilon - \epsilon\lambda}{2+2\lambda}\right)^d\right) 
        = \frac{\log\left(\frac{2^d}{d!}(\frac{\epsilon - \epsilon\lambda}{2+2\lambda})^d\right)}{\log\left(\frac{d}{d+1}\right)}
        =  \frac{\log\left(\frac{d!}{2^d}(\frac{2}{\epsilon}\frac{1 + \lambda}{1 - \lambda})^d\right)}{\log\left(\frac{d + 1}{d}\right)}
    \]
    \[
        \leq \frac{\log d! - d + d\left(\log 4 + \log (\frac{1}{\epsilon}) + \log (\frac{1}{1-\lambda})\right)}{\log(\frac{d+1}{d})}
    \]
    \[
        \leq \frac{\bigO \left( d \left(\log d + \log(\frac{1}{\epsilon}) + \log(\frac{1}{1 - \lambda})\right) \right)}{\log(1+\frac{1}{d})}.
    \]
    For small $\delta > 0$, we have $\log(1+\delta)\approx \delta$, and thus we get $k \leq \bigO(d^2(\log d + \log(\frac{1}{\epsilon}) + \log(\frac{1}{1-\lambda}))$. 
    To get rid of the $\log d$ term, observe that if $\max (\frac{1}{ \epsilon}, \frac{1}{1-\lambda}) < d$ (i.e., whenever this term matters), a simple iteration algorithm can find an $\epsilon$-approximate fixpoint after at most $\bigO(d \log d)$ queries. In particular, let $x^{(0)} \in [0, 1]^d$ be arbitrary and consider the recursively defined iterates $x^{(i)} \coloneqq f(x^{(i - 1)})$ for $i \geq 1$. Since $f$ is contracting, we have $\norm{x^{(i)} - f(x^{(i)})}_p \leq \lambda \norm{x^{(i - 1)} - f(x^{(i - 1)})}_p$ for all $i \geq 1$. Together with the observation $\norm{x^{(0)} - f(x^{(0)})}_p \leq d$, we get $\norm{x^{(i)} - f(x^{(i)})}_p \leq \lambda^i d$ for all $i$. If $i$ is large enough such that $\lambda^i d < \epsilon$, we must have found an $\epsilon$-approximate fixpoint. This is equivalent to $i > \frac{\log d / \epsilon}{\log \frac{1}{\lambda}}$. By using $\max (\frac{1}{ \epsilon}, \frac{1}{1-\lambda}) < d$, we get
    \[
        \log( \frac{1}{\lambda} ) = \log ( 1 + \frac{1 - \lambda}{\lambda}) \geq 
        \log ( 1 + \frac{1}{d\lambda}) \geq \log ( 1 + \frac{1}{d})
    \]
    and hence we obtain that $\bigO(d \log d)$ iterations suffice.
\end{proof}

\subsection{Rounding to the Grid in the \texorpdfstring{$\ell_1$}{l1}-Case}\label{ssec:discrete_algorithms}

In this section, we adapt the algorithm from the previous section to also work in the discretized setting \gridToGrid{1}, proving the following theorem.

\begin{theorem}\label{thm:mainGridtogrid}
    For every $b \geq \log_2 \left(\frac{2d}{\epsilon} \frac{1 + \lambda}{1 - \lambda} \right)$, an $\epsilon$-approximate fixpoint of a $\lambda$-contracting (in $\ell_1$-norm) grid map $f:G^d_b\rightarrow [0,1]^d$ can be found using $\bigO(d^2 ( \log\frac{1}{\epsilon} + \log\frac{1}{1-\lambda}))$ queries.
\end{theorem}

The main issue that we have to address is that we cannot always query the centerpoint $c$ of the remaining search space, since $c$ is not guaranteed to lie on the grid $G^d_b$. In the $\ell_\infty$-case, Chen, Li, and Yannakakis~\cite{chenComputingFixedPoint2024} solved this problem by choosing the measure for the remaining search space appropriately. Concretely, they design a special measure to determine the size of the remaining search space, and then show that rounding the centerpoint to the grid works for this measure. We do not repeat their argument here, but note that their technique also works if applied together with our (discrete) centerpoint theorem.

Instead, we focus on the $\ell_1$-case. Fortunately, rounding in the $\ell_1$-case turns out to be simpler than in the $\ell_\infty$-case. Indeed, we simply measure the size of the remaining search space as the number of grid points that still remain inside. It turns out that appropriately rounding the centerpoint to the grid does not change its centerpoint properties with respect to this measure of size. More formally, we get the following lemma.

\begin{lemma}
\label{lemma:1-norm_rounding}
    Consider an arbitrary subset of points $P \subseteq G^d_b$ on the grid. There exists a point $c \in G^d_b$ such that $|\limitH^1_{c, v} \cap P| \geq \frac{|P|}{d + 1}$ for all $v \in S^{d - 1}$.
\end{lemma}

We postpone the proof of \Cref{lemma:1-norm_rounding} and instead first use it to derive \Cref{thm:mainGridtogrid}.

\begin{proof}[Proof of Theorem~\ref{thm:mainGridtogrid}]
Overall, we use a similar strategy as in the proof of \Cref{thm:mainContinuous}, but we use \Cref{lemma:1-norm_rounding} that is derived from our discrete centerpoint theorem. Concretely, let $M \subseteq G^d_b$ be the subset of grid points that could still be $\epsilon$-approximate fixpoints (they have not been ruled out yet). Applying \Cref{lemma:1-norm_rounding} to $M$ yields a centerpoint $c$ on the grid $G^d_b$. Querying $c$ guarantees that we can exclude at least $\frac{|M|}{d + 1}$ points from our current search space in each non-terminating iteration. 

Now \Cref{lem:ballaroundfixpointstays} guarantees that as long as we have not queried an $\epsilon$-approximate fixpoint, no point in $B^1(x^\star, r_{\epsilon, \lambda}) \cap G^d_b$ has been excluded from the search space yet (where $r_{\epsilon, \lambda} = \frac{\epsilon - \epsilon \lambda}{2+2\lambda}$). To guarantee that this intersection is non-empty, it suffices to have $2^{-b} \leq \frac{r_{\epsilon, \lambda}}{d}$, which translates to our assumption $b \geq \log_2 \left(\frac{2d}{\epsilon} \frac{1 + \lambda}{1 - \lambda} \right)$. In fact, we can assume 
\[
    b = \left\lceil \log_2 \left(\frac{2d}{\epsilon} \frac{1 + \lambda}{1 - \lambda} \right) \right\rceil 
\]
without loss of generality (as we can use a subgrid of this size, if the given grid is finer). For this choice of $b$, we have $|G^d_b| = (2^b + 1)^d \leq 4^{db}$. With the calculation
\[
    \log_{\frac{d}{d + 1}} 4^{-db} = \frac{\log  4^{-db} }{\log \frac{d}{d + 1}} = \frac{db}{\log (1 + \frac{1}{d} )} \leq \bigO(d^2 b) \leq \bigO \left(d^2 \left(\log d + \log \frac{1}{\epsilon} + \log \frac{1}{1 - \lambda} \right) \right)
\]
we therefore conclude that the algorithm must find an $\epsilon$-approximate fixpoint after at most $\bigO(d^2 ( \log\frac{1}{\epsilon} + \log\frac{1}{1-\lambda} + \log d))$ queries to grid points. 

As in the proof of \Cref{thm:mainContinuous}, we can get rid of the $\log d$ term in our query bound by observing that if $\max (\frac{1}{ \epsilon}, \frac{1}{1-\lambda}) < d$, a simple iteration algorithm (this time with rounding) can find an $\epsilon$-approximate fixpoint after at most $\bigO(d \log d)$ queries. 
\end{proof}

It remains to prove \Cref{lemma:1-norm_rounding}. For this, we need some additional theory. In particular, we will proceed to give a characterization of containment in an $\ell_1$-halfspace based on tools from convex analysis. Once this characterization is established, it will be quite easy to derive \Cref{lemma:1-norm_rounding}. Note that more details on this characterization can be found in \Cref{ssec:properties_I}.

The main tool that we need is the notion of subgradients of a convex function $f : \R^d \rightarrow \R$: a vector $u \in \R^d$ is a subgradient of $f$ at $x \in \R^d$ if and only if 
\[
    f(x') - f(x) \geq \langle u, (x' - x) \rangle
\]
for all $x' \in \R^d$. The set $\partial f(x) \subseteq \R^d$ of all subgradients of $f$ at $x$ is also called the subdifferential. If $f$ is differentiable at $x$, then $\partial f(x) = \{\nabla f(x)\}$. For more details on these concepts and convex analysis in general, we refer to the standard textbook by Rockafellar~\cite{rockafellarConvexAnalysis1970}.

Subgradients are useful for us due to the following characterization of containment for $\ell_p$-halfspaces. The proof can be found in \Cref{ssec:properties_I}.

\begin{restatable}{lemma}{lemmasubgradientscontainment}
\label{lemma:subgradients_and_containment}\RestateRemark
    For any $p \in [1, \infty) \cup \{\infty\}$, a point $z \in \R^d$ is contained in an $\ell_p$-halfspace $\limitH^p_{x, v}$ if and only if there exists a subgradient $u \in \partial ||z -x||_p$ with $\langle u, v \rangle \geq 0$.
\end{restatable}

In this section, we are interested in applying \Cref{lemma:subgradients_and_containment} in the $\ell_1$-case. Therefore, let us next describe the subgradients of $||\cdot||_1$. Due to \Cref{lem:structureoflimithalfspace} ($\ell_p$-halfspaces are unions of rays), it will suffice to characterize the subgradients at points $z$ with $||z||_1 = 1$.

\begin{restatable}{observation}{obssubdifferentialmanhattan}
\label{obs:l_1_subdifferential}\RestateRemark
    Consider arbitrary $z \in \R^d$ with $||z||_1 = 1$. A vector $u \in \R^d$ is a subgradient of $||\cdot||_1$ at $z$ if and only if 
    \[
        u_i \in \begin{cases}
            \{1\} & \text{ if } z_i > 0 \\
            \{-1\} & \text{ if } z_i < 0 \\
            [-1, 1] & \text{ if } z_i = 0
        \end{cases}
    \]
    for all $i \in [d]$. In particular, we have $\langle u, z \rangle = ||u||_{\infty} = ||z||_1 = 1$ for all $u \in \partial ||z||_1$.
\end{restatable}

Equipped with \Cref{lemma:subgradients_and_containment} and \Cref{obs:l_1_subdifferential}, we are now ready to prove \Cref{lemma:1-norm_rounding}. 

\begin{proof}[Proof of Theorem~\ref{lemma:1-norm_rounding}]
    Let $P \subseteq G^d_b$ be an arbitrary subset of the grid points, as given in the lemma. By \Cref{theorem:discrete_p-centerpoint}, there exists a discrete $\ell_1$-centerpoint $c$ of $P$. Moreover, by \Cref{obs:centerpoint_in_bounding_box}, $c$ is contained in the convex hull $\conv(G^d_b) = [0, 1]^d$ of the grid. Let $c' \in G^d_b$ be the grid point closest to $c$ (in $\ell_1$-distance). In other words, $c'$ is obtained from $c$ by rounding each coordinate individually as little as possible. 
    
    We claim that $c'$ is also a discrete $\ell_1$-centerpoint of $P$. To prove this, consider the $\ell_1$-halfspaces $\limitH^1_{c', v}$ and $\limitH^1_{c, v}$ for an arbitrary direction $v \in S^{d - 1}$. We prove that every point $z \in P \cap \limitH^1_{c, v}$ must also be contained in $\limitH^1_{c', v}$. Indeed, by \Cref{lemma:subgradients_and_containment}, there exists a subgradient $u \in \partial ||z - c||_1$ such that $\langle u, v \rangle \geq 0$. We want to prove $u \in \partial ||z - c'||_1$. Observe that our choice of $c'$ guarantees 
    \[
        z_i - c'_i > 0 \implies z_i - c_i > 0 \text{ and } z_i - c'_i < 0 \implies z_i - c_i < 0
    \]
    for all $i \in [d]$. With \Cref{obs:l_1_subdifferential}, we therefore have $u \in \partial ||z - c'||_1$ and conclude by \Cref{lemma:subgradients_and_containment} that $z$ is contained in $\limitH^1_{c', v}$.
\end{proof}

\subsection{Total Search Version}\label{sec:totalproblem}\label{ssec:total_search_problem}

The class \TFNPdt (standing for \emph{Total Function NP Search Problems, decision-tree view}), as defined by Göös et al.~\cite{goosSeparationsProofComplexity2024}, captures total search problems (i.e., problems where every possible instance has a solution) that are specified by a long hidden bitstring accessible only through a bit querying oracle. To lie in \TFNPdt, solutions must be efficiently verifiable (i.e., by decision trees of depth poly-logarithmic in the length of the bitstring). To lie in \FPdt, the subclass of efficiently solvable problems in \TFNPdt, one must be able to efficiently \emph{find} a solution as well, i.e., there has to be a decision tree of poly-logarithmic depth that always outputs a correct solution.

To fit a \emph{promise} search problem into \TFNPdt, we must introduce solution types, usually called \emph{violations}, that are guaranteed to exist when the promise is violated.

Chen, Li, and Yannakakis~\cite{chenComputingFixedPoint2024} showed how the \gridToGrid{\infty} problem can be made total: under the $\ell_\infty$-norm, a grid-map $f:G^d_b\rightarrow [0,1]^d$ extends to a $\lambda$-contraction $f'$ on $[0,1]^d$ if and only if $f$ is $\lambda$-contracting for all pairs of points in $G^d_b$. Thus, it suffices to introduce violations consisting of two points $x,y \in G^d_b$ for which $\norm{f(x)-f(y)}_\infty > \lambda\cdot \norm{x-y}_\infty$. If a function $f:G^d_b\rightarrow G^d_{b'}$ (for $b'\in\poly(b)$) is now encoded as a bitstring of length $2^{d\cdot b}\cdot b'$ by simply concatenating the output values $f(x)$ for all $x\in G^d_b$, both $\epsilon$-approximate fixpoints as well as these violations can be verified by querying $\bigO(b')\in \poly\log(2^{d\cdot b}\cdot b')$ bits, thus placing the resulting problem in \TFNPdt. Furthermore, the algorithm of Chen, Li, and Yannakakis can also be used to solve this total version of the problem: if the function $f$ is $\lambda$-contracting for all pairs of points queried by the algorithm, the algorithm must return an $\epsilon$-approximate fixpoint. Otherwise, the algorithm must have encountered a violation. Thus, \gridToGrid{\infty} is in \FPdt.

The same strategy does not seem to work in the $\ell_1$-case. In fact, we suspect that the statement that a grid-map extends to a contraction map if and only if it is contracting for all grid points is not true for the $\ell_1$-norm and general $\ell_p$-norms. In the case of $p=\infty$, Chen, Li, and Yannakakis constructed an extension explicitly: they extended the grid-function to $[0,1]^d$ using a formula that corresponds to applying McShane's extension lemma~\cite{mcshaneExtensionRangeFunctions1934} to every partial function $f_i:G^d_b\rightarrow [0,1]$ independently. McShane's extension lemma guarantees that the partial functions $f'_i:[0,1]^d\rightarrow [0,1]$ then fulfill the same contraction property as $f_i$, which together with the definition of the $\ell_\infty$-norm as a maximum over coordinates, yields the contraction property for $f'$. Trying this for $p=1$, the analysis only yields that $f'$ is $(d\lambda)$-Lipschitz, which is not enough to be contracting if we for example have $d>1$ and $\lambda>0.5$.

We get around this in a naive way by introducing a violation type based on the termination criterion of our algorithm for \gridToGrid{1}. Concretely, if our algorithm does not find an $\epsilon$-approximate fixpoint after reaching its query bound, then the queried points $x^{(1)}, \dots, x^{(k)}$ must satisfy $G^d_b \setminus \bigcup_{i \in [k]} H^1_{x^{(i)}, f(x^{(i)})} = \emptyset$. In other words, we certify the existence of a violation of the contraction property with a set of points on the grid whose associated bisector halfspaces contain all grid points in their union.

\begin{definition}
\label{def:total_search_problem}
    An instance of the \totalContraction{1} problem consists of a bitstring encoding integers $d,b,b'\in \N$ and integral logarithms $\log(\frac{1}{1-\lambda}),\log(\frac{1}{\epsilon})$ in unary, as well as a function $f:G^d_b\rightarrow G^d_{b'}$ encoded as a concatenation of its values. The goal is to produce one of the following. \begin{itemize}
        \item[(S)] A point $x\in G^d_b$ such that $\norm{f(x)-x}_p\leq\epsilon$.
        \item[(V)] A set $P$ of $\poly(b, d, \log \left(\frac{1}{\epsilon}\right), \log \left(\frac{1}{1 - \lambda} \right))$ points in $G^d_b$ with $G^d_b \subseteq \bigcup_{x \in P} H^1_{x, f(x)}$.
    \end{itemize}
\end{definition}

Note that for our algorithm to work, we need to syntactically guarantee $b \geq \log_2 \left(\frac{2d}{\epsilon} \frac{1 + \lambda}{1 - \lambda} \right)$. In fact, without this assumption, we would not even know if the problem is total.  We also syntactically guarantee $b'\in\poly(b)$ to make sure that the number of bit queries that our algorithm makes to the encoding of $f$ is polynomially bounded. Such syntactic guarantees can be realized by adding trivial-to-verify additional violation types for the cases where the input parameters do not fulfill these assumptions.

We have shown in the proof of \Cref{thm:mainGridtogrid} that after $\bigO(d^2 ( \log\frac{1}{\epsilon} + \log\frac{1}{1-\lambda} + \log d))$ queries to $f$, the queries made by our algorithm must have produced either a solution (S), or a violation (V). Thus, our algorithm solves \totalContraction{1} as well, placing it in \TFNPdt and even in \FPdt.

\begin{corollary}
    $\totalContraction{1} \in\FPdt$.
\end{corollary}

Note that for general $p \notin \{1, 2, \infty\}$, it is unclear whether an analogue of  \Cref{def:total_search_problem} would even yield a total problem: we do not provide any rounding strategy for those values and hence it is unclear whether any grid-map without $\epsilon$-approximate fixpoints on the grid must have a certificate for this in terms of a violation (V) as described above.

\printbibliography

@inproceedings{DBLP:conf/focs/HaslebacherLSW25,
  author       = {Sebastian Haslebacher and
                  Jonas Lill and
                  Patrick Schnider and
                  Simon Weber},
  title        = {Query-Efficient Fixpoints of {$\ell_p$}-Contractions},
  booktitle    = {66th {IEEE} Annual Symposium on Foundations of Computer Science, {FOCS}
                  2025, Sydney, Australia, December 14-17, 2025},
  pages        = {2037--2053},
  publisher    = {{IEEE}},
  year         = {2025},
  url          = {https://doi.org/10.1109/FOCS63196.2025.00109},
  doi          = {10.1109/FOCS63196.2025.00109},
  bibsource    = {dblp computer science bibliography, https://dblp.org}
}

@article {banach1922operations,
    AUTHOR = {Banach, Stefan},
     TITLE = {Sur les op\'erations dans les ensembles abstraits et leur
              application aux \'equations int\'egrales},
   JOURNAL = {Fund. Math.},
  FJOURNAL = {Polska Akademia Nauk. Fundamenta Mathematicae},
    VOLUME = {3},
      YEAR = {1922},
     PAGES = {133--181},
      ISSN = {0016-2736,1730-6329},
   MRCLASS = {46-02 (47-02)},
  MRNUMBER = {3949898},
       DOI = {10.4064/fm-3-1-133-181},
       URL = {https://doi.org/10.4064/fm-3-1-133-181},
}

@article {basit2010centerdisks,
    AUTHOR = {Basit, Abdul and Mustafa, Nabil H. and Ray, Saurabh and Raza,
              Sarfraz},
     TITLE = {Centerpoints and {T}verberg's technique},
   JOURNAL = {Comput. Geom.},
  FJOURNAL = {Computational Geometry. Theory and Applications},
    VOLUME = {43},
      YEAR = {2010},
    NUMBER = {6-7},
     PAGES = {593--600},
      ISSN = {0925-7721,1879-081X},
   MRCLASS = {52B55},
  MRNUMBER = {2639228},
       DOI = {10.1016/j.comgeo.2010.03.002},
       URL = {https://doi.org/10.1016/j.comgeo.2010.03.002},
}

@article {brouwerUeberAbbildungMannigfaltigkeiten1911,
    AUTHOR = {Brouwer, L. E. J.},
     TITLE = {{\"U}ber {A}bbildung von {M}annigfaltigkeiten},
   JOURNAL = {Math. Ann.},
  FJOURNAL = {Mathematische Annalen},
    VOLUME = {71},
      YEAR = {1911},
    NUMBER = {1},
     PAGES = {97--115},
      ISSN = {0025-5831,1432-1807},
   MRCLASS = {99-04},
  MRNUMBER = {1511644},
       DOI = {10.1007/BF01456931},
       URL = {https://doi.org/10.1007/BF01456931},
}

@article{chang2010lowerbounds,
  author       = {Ching{-}Lueh Chang and
                  Yuh{-}Dauh Lyuu},
  title        = {Optimal bounds on finding fixed points of contraction mappings},
  journal      = {Theor. Comput. Sci.},
  volume       = {411},
  number       = {16-18},
  pages        = {1742--1749},
  year         = {2010},
  url          = {https://doi.org/10.1016/j.tcs.2010.01.016},
  doi          = {10.1016/J.TCS.2010.01.016},
  bibsource    = {dblp computer science bibliography, https://dblp.org}
}

@article{chenComputingFixedPoint2024,
  author       = {Xi Chen and
                  Yuhao Li and
                  Mihalis Yannakakis},
  title        = {Computing a Fixed Point of Contraction Maps in Polynomial Queries},
  journal      = {J. {ACM}},
  volume       = {72},
  number       = {4},
  pages        = {25:1--25:20},
  year         = {2025},
  url          = {https://doi.org/10.1145/3744738},
  doi          = {10.1145/3744738},
  bibsource    = {dblp computer science bibliography, https://dblp.org}
}

@inproceedings{cherapanamjeri2024approximatecenterpoints,
  author       = {Yeshwanth Cherapanamjeri},
  title        = {Computing Approximate Centerpoints in Polynomial Time},
  booktitle    = {65th {IEEE} Annual Symposium on Foundations of Computer Science, {FOCS}
                  2024, Chicago, IL, USA, October 27-30, 2024},
  pages        = {1654--1668},
  publisher    = {{IEEE}},
  year         = {2024},
  url          = {https://doi.org/10.1109/FOCS61266.2024.00104},
  doi          = {10.1109/FOCS61266.2024.00104},
  bibsource    = {dblp computer science bibliography, https://dblp.org}
}

@article{condonComplexityStochasticGames1992,
  author       = {Anne Condon},
  title        = {The Complexity of Stochastic Games},
  journal      = {Inf. Comput.},
  volume       = {96},
  number       = {2},
  pages        = {203--224},
  year         = {1992},
  url          = {https://doi.org/10.1016/0890-5401(92)90048-K},
  doi          = {10.1016/0890-5401(92)90048-K},
  bibsource    = {dblp computer science bibliography, https://dblp.org}
}

@inproceedings{daskalakisContinuousLocalSearch2011,
  author       = {Constantinos Daskalakis and
                  Christos H. Papadimitriou},
  editor       = {Dana Randall},
  title        = {Continuous Local Search},
  booktitle    = {Proceedings of the Twenty-Second Annual {ACM-SIAM} Symposium on Discrete
                  Algorithms, {SODA} 2011, San Francisco, California, USA, January 23-25,
                  2011},
  pages        = {790--804},
  publisher    = {{SIAM}},
  year         = {2011},
  url          = {https://doi.org/10.1137/1.9781611973082.62},
  doi          = {10.1137/1.9781611973082.62},
  bibsource    = {dblp computer science bibliography, https://dblp.org}
}

@inproceedings{daskalakisConverseBanachFixed2018,
  author       = {Constantinos Daskalakis and
                  Christos Tzamos and
                  Manolis Zampetakis},
  editor       = {Ilias Diakonikolas and
                  David Kempe and
                  Monika Henzinger},
  title        = {A converse to {B}anach's fixed point theorem and its {CLS}-completeness},
  booktitle    = {Proceedings of the 50th Annual {ACM} {SIGACT} Symposium on Theory
                  of Computing, {STOC} 2018, Los Angeles, CA, USA, June 25-29, 2018},
  pages        = {44--50},
  publisher    = {{ACM}},
  year         = {2018},
  url          = {https://doi.org/10.1145/3188745.3188968},
  doi          = {10.1145/3188745.3188968},
  bibsource    = {dblp computer science bibliography, https://dblp.org}
}

@incollection{dohrauARRIVALZeroPlayerGraph2017,
  title = {{{ARRIVAL}}: {{A Zero-Player Graph Game}} in {{NP}} {$\cap$} {{coNP}}},
  shorttitle = {{{ARRIVAL}}},
  booktitle = {A {{Journey Through Discrete Mathematics}}: {{A Tribute}} to {{Ji{\v r}{\'i} Matou{\v s}ek}}},
  author = {Dohrau, J{\'e}r{\^o}me and G{\"a}rtner, Bernd and Kohler, Manuel and Matou{\v s}ek, Ji{\v r}{\'i} and Welzl, Emo},
  year = {2017},
  pages = {367--374},
  publisher = {Springer International Publishing},
  address = {Cham},
  doi = {10.1007/978-3-319-44479-6_14},
  urldate = {2023-11-13},
  isbn = {978-3-319-44479-6},
  langid = {english}
}

@inproceedings{etessamiTarskiTheoremSupermodular2020,
  author       = {Kousha Etessami and
                  Christos H. Papadimitriou and
                  Aviad Rubinstein and
                  Mihalis Yannakakis},
  editor       = {Thomas Vidick},
  title        = {Tarski's Theorem, Supermodular Games, and the Complexity of Equilibria},
  booktitle    = {11th Innovations in Theoretical Computer Science Conference, {ITCS}
                  2020, Seattle, Washington, USA, January 12-14, 2020},
  series       = {LIPIcs},
  volume       = {151},
  pages        = {18:1--18:19},
  publisher    = {Schloss Dagstuhl - Leibniz-Zentrum f{\"{u}}r Informatik},
  year         = {2020},
  url          = {https://doi.org/10.4230/LIPIcs.ITCS.2020.18},
  doi          = {10.4230/LIPICS.ITCS.2020.18},
  bibsource    = {dblp computer science bibliography, https://dblp.org}
}

@article{fearnleyComplexityGradientDescent2022,
  author       = {John Fearnley and
                  Paul Goldberg and
                  Alexandros Hollender and
                  Rahul Savani},
  title        = {The Complexity of Gradient Descent: {CLS} = {PPAD} {\(\cap\)} {PLS}},
  journal      = {J. {ACM}},
  volume       = {70},
  number       = {1},
  pages        = {7:1--7:74},
  year         = {2023},
  url          = {https://doi.org/10.1145/3568163},
  doi          = {10.1145/3568163},
  bibsource    = {dblp computer science bibliography, https://dblp.org}
}

@article{fearnleyUniqueEndPotential2020,
  author       = {John Fearnley and
                  Spencer Gordon and
                  Ruta Mehta and
                  Rahul Savani},
  title        = {Unique end of potential line},
  journal      = {J. Comput. Syst. Sci.},
  volume       = {114},
  pages        = {1--35},
  year         = {2020},
  url          = {https://doi.org/10.1016/j.jcss.2020.05.007},
  doi          = {10.1016/J.JCSS.2020.05.007},
  bibsource    = {dblp computer science bibliography, https://dblp.org}
}

@article{goosFurtherCollapsesTFNP2024,
  author       = {Mika G{\"{o}}{\"{o}}s and
                  Alexandros Hollender and
                  Siddhartha Jain and
                  Gilbert Maystre and
                  William Pires and
                  Robert Robere and
                  Ran Tao},
  title        = {Further Collapses in {TFNP}},
  journal      = {{SIAM} J. Comput.},
  volume       = {53},
  number       = {3},
  pages        = {573--587},
  year         = {2024},
  url          = {https://doi.org/10.1137/22m1498346},
  doi          = {10.1137/22M1498346},
  bibsource    = {dblp computer science bibliography, https://dblp.org}
}

@article{goosSeparationsProofComplexity2024,
  author       = {Mika G{\"{o}}{\"{o}}s and
                  Alexandros Hollender and
                  Siddhartha Jain and
                  Gilbert Maystre and
                  William Pires and
                  Robert Robere and
                  Ran Tao},
  title        = {Separations in Proof Complexity and {TFNP}},
  journal      = {J. {ACM}},
  volume       = {71},
  number       = {4},
  pages        = {26:1--26:45},
  year         = {2024},
  url          = {https://doi.org/10.1145/3663758},
  doi          = {10.1145/3663758},
  bibsource    = {dblp computer science bibliography, https://dblp.org}
}

@article {grunbaumPartitionsMassdistributionsConvex1960,
    AUTHOR = {Gr\"unbaum, B.},
     TITLE = {Partitions of mass-distributions and of convex bodies by
              hyperplanes},
   JOURNAL = {Pacific J. Math.},
  FJOURNAL = {Pacific Journal of Mathematics},
    VOLUME = {10},
      YEAR = {1960},
     PAGES = {1257--1261},
      ISSN = {0030-8730,1945-5844},
   MRCLASS = {53.90 (52.40)},
  MRNUMBER = {124818},
MRREVIEWER = {H.\ G.\ Eggleston},
       URL = {http://projecteuclid.org/euclid.pjm/1103038065},
}

@inproceedings{haslebacherARRIVALRecursiveFramework2025,
  author       = {Sebastian Haslebacher},
  editor       = {Keren Censor{-}Hillel and
                  Fabrizio Grandoni and
                  Jo{\"{e}}l Ouaknine and
                  Gabriele Puppis},
  title        = {{{ARRIVAL}}: {{Recursive Framework}} \& {$\ell_1$}-{{Contraction}}},
  booktitle    = {52nd International Colloquium on Automata, Languages, and Programming,
                  {ICALP} 2025, Aarhus, Denmark, July 8-11, 2025},
  series       = {LIPIcs},
  volume       = {334},
  pages        = {95:1--95:17},
  publisher    = {Schloss Dagstuhl - Leibniz-Zentrum f{\"{u}}r Informatik},
  year         = {2025},
  url          = {https://doi.org/10.4230/LIPIcs.ICALP.2025.95},
  doi          = {10.4230/LIPICS.ICALP.2025.95},
  bibsource    = {dblp computer science bibliography, https://dblp.org}
}

@article{Helly1923,
  title = {{\"U}ber {{Mengen}} Konvexer {{K{\"o}rper}} Mit Gemeinschaftlichen {{Punkten}}.},
  author = {Helly, Eduard},
  year = {1923},
  journal = {Jahresbericht der Deutschen Mathematiker-Vereinigung},
  volume = {32},
  pages = {175--176}
}

@article{huangApproximatingFixedPoints1999,
  author       = {Z. Huang and
                  Leonid G. Khachiyan and
                  Christopher (Krzysztof) Sikorski},
  title        = {Approximating Fixed Points of Weakly Contracting Mappings},
  journal      = {J. Complex.},
  volume       = {15},
  number       = {2},
  pages        = {200--213},
  year         = {1999},
  url          = {https://doi.org/10.1006/jcom.1999.0504},
  doi          = {10.1006/JCOM.1999.0504},
  bibsource    = {dblp computer science bibliography, https://dblp.org}
}

@phdthesis{ishizuka2022tfnpFixpoints,
  title = {On {{TFNP}} Classes: {{Approaches}} from Fixed Point Theory and Algorithmic Game Theory},
  author = {Ishizuka, Takashi},
  year = {2022},
  school = {Kyushu University},
  url = {https://hdl.handle.net/2324/5068170}
}

@article{karasev2014projectivecenterpoints,
  author       = {Roman N. Karasev and
                  Benjamin Matschke},
  title        = {Projective Center Point and {T}verberg Theorems},
  journal      = {Discret. Comput. Geom.},
  volume       = {52},
  number       = {1},
  pages        = {88--101},
  year         = {2014},
  url          = {https://doi.org/10.1007/s00454-014-9602-9},
  doi          = {10.1007/S00454-014-9602-9},
  bibsource    = {dblp computer science bibliography, https://dblp.org}
}

@article{li2000learning,
  author       = {Yi Li and
                  Philip M. Long and
                  Aravind Srinivasan},
  title        = {Improved Bounds on the Sample Complexity of Learning},
  journal      = {J. Comput. Syst. Sci.},
  volume       = {62},
  number       = {3},
  pages        = {516--527},
  year         = {2001},
  url          = {https://doi.org/10.1006/jcss.2000.1741},
  doi          = {10.1006/JCSS.2000.1741},
  bibsource    = {dblp computer science bibliography, https://dblp.org}
}

@article {deloera2019mentionofbrouwer,
    AUTHOR = {De Loera, Jes\'us A. and Goaoc, Xavier and Meunier,
              Fr\'ed\'eric and Mustafa, Nabil H.},
     TITLE = {The discrete yet ubiquitous theorems of {C}arath\'eodory,
              {H}elly, {S}perner, {T}ucker, and {T}verberg},
   JOURNAL = {Bull. Amer. Math. Soc. (N.S.)},
  FJOURNAL = {American Mathematical Society. Bulletin. New Series},
    VOLUME = {56},
      YEAR = {2019},
    NUMBER = {3},
     PAGES = {415--511},
      ISSN = {0273-0979,1088-9485},
   MRCLASS = {52A35 (57M99 90C25 91A80 91B32)},
  MRNUMBER = {3974609},
MRREVIEWER = {Mircea\ Balaj},
       DOI = {10.1090/bull/1653},
       URL = {https://doi.org/10.1090/bull/1653},
}

@inproceedings{matousek1991centerpoint,
  author       = {Jir{\'{\i}} Matousek},
  editor       = {Jacob E. Goodman and
                  Richard Pollack and
                  William Steiger},
  title        = {Computing the Center of Planar Point Sets},
  booktitle    = {Discrete and Computational Geometry: Papers from the {DIMACS} Special
                  Year},
  series       = {{DIMACS} Series in Discrete Mathematics and Theoretical Computer Science},
  volume       = {6},
  pages        = {221--230},
  publisher    = {{DIMACS/AMS}},
  year         = {1990},
  url          = {https://doi.org/10.1090/dimacs/006/14},
  doi          = {10.1090/DIMACS/006/14},
  bibsource    = {dblp computer science bibliography, https://dblp.org}
}

@article {mcshaneExtensionRangeFunctions1934,
    AUTHOR = {McShane, Edward J.},
     TITLE = {Extension of range of functions},
   JOURNAL = {Bull. Amer. Math. Soc.},
  FJOURNAL = {Bulletin of the American Mathematical Society},
    VOLUME = {40},
      YEAR = {1934},
    NUMBER = {12},
     PAGES = {837--842},
      ISSN = {0002-9904},
   MRCLASS = {99-04},
  MRNUMBER = {1562984},
       DOI = {10.1090/S0002-9904-1934-05978-0},
       URL = {https://doi.org/10.1090/S0002-9904-1934-05978-0},
}

@inproceedings{millerApproximateCenterPoints2009,
  author       = {Gary L. Miller and
                  Donald R. Sheehy},
  editor       = {John Hershberger and
                  Efi Fogel},
  title        = {Approximate center points with proofs},
  booktitle    = {Proceedings of the 25th {ACM} Symposium on Computational Geometry ({SCG}'09),
                  Aarhus, Denmark, June 8-10, 2009},
  pages        = {153--158},
  publisher    = {{ACM}},
  year         = {2009},
  url          = {https://doi.org/10.1145/1542362.1542395},
  doi          = {10.1145/1542362.1542395},
  bibsource    = {dblp computer science bibliography, https://dblp.org}
}

@incollection{mustafa2017handbook,
  title = {Epsilon-Approximations \& Epsilon-Nets},
  booktitle = {Handbook of Discrete and Computational Geometry},
  author = {Mustafa, Nabil H. and Varadarajan, Kasturi},
  year = {2017},
  edition = {3},
  pages = {1241--1267},
  publisher = {CRC Press LLC},
  chapter = {47}
}

@article{papadimitriouComplexityParityArgument1994,
  author       = {Christos H. Papadimitriou},
  title        = {On the Complexity of the Parity Argument and Other Inefficient Proofs
                  of Existence},
  journal      = {J. Comput. Syst. Sci.},
  volume       = {48},
  number       = {3},
  pages        = {498--532},
  year         = {1994},
  url          = {https://doi.org/10.1016/S0022-0000(05)80063-7},
  doi          = {10.1016/S0022-0000(05)80063-7},
  bibsource    = {dblp computer science bibliography, https://dblp.org}
}

@inproceedings{pilz2018multiplecenterpoints,
  author       = {Alexander Pilz and
                  Patrick Schnider},
  editor       = {Wen{-}Lian Hsu and
                  Der{-}Tsai Lee and
                  Chung{-}Shou Liao},
  title        = {Extending the Centerpoint Theorem to Multiple Points},
  booktitle    = {29th International Symposium on Algorithms and Computation, {ISAAC}
                  2018, Jiaoxi, Yilan, Taiwan, December 16-19, 2018},
  series       = {LIPIcs},
  volume       = {123},
  pages        = {53:1--53:13},
  publisher    = {Schloss Dagstuhl - Leibniz-Zentrum f{\"{u}}r Informatik},
  year         = {2018},
  url          = {https://doi.org/10.4230/LIPIcs.ISAAC.2018.53},
  doi          = {10.4230/LIPICS.ISAAC.2018.53},
  bibsource    = {dblp computer science bibliography, https://dblp.org}
}

@inproceedings{rademacherApproximatingCentroidHard2007,
  author       = {Luis Rademacher},
  editor       = {Jeff Erickson},
  title        = {Approximating the centroid is hard},
  booktitle    = {Proceedings of the 23rd {ACM} Symposium on Computational Geometry ({SCG}'07),
                  Gyeongju, South Korea, June 6-8, 2007},
  pages        = {302--305},
  publisher    = {{ACM}},
  year         = {2007},
  url          = {https://doi.org/10.1145/1247069.1247123},
  doi          = {10.1145/1247069.1247123},
  bibsource    = {dblp computer science bibliography, https://dblp.org}
}

@article {radoTheoremGeneralMeasure1946,
    AUTHOR = {Rado, R.},
     TITLE = {A theorem on general measure},
   JOURNAL = {J. London Math. Soc.},
  FJOURNAL = {The Journal of the London Mathematical Society},
    VOLUME = {21},
      YEAR = {1946},
     PAGES = {291--300 (1947)},
      ISSN = {0024-6107,1469-7750},
   MRCLASS = {27.2X},
  MRNUMBER = {21962},
MRREVIEWER = {H.\ Federer},
       DOI = {10.1112/jlms/s1-21.4.291},
       URL = {https://doi.org/10.1112/jlms/s1-21.4.291},
}

@article {radon1921lemma,
    AUTHOR = {Radon, Johann},
     TITLE = {Mengen konvexer {K}\"orper, die einen gemeinsamen {P}unkt
              enthalten},
   JOURNAL = {Math. Ann.},
  FJOURNAL = {Mathematische Annalen},
    VOLUME = {83},
      YEAR = {1921},
    NUMBER = {1-2},
     PAGES = {113--115},
      ISSN = {0025-5831,1432-1807},
   MRCLASS = {99-04},
  MRNUMBER = {1512002},
       DOI = {10.1007/BF01464231},
       URL = {https://doi.org/10.1007/BF01464231},
}

@inproceedings{ray2007generalizedcenterpoint,
  author       = {Saurabh Ray and
                  Nabil H. Mustafa},
  editor       = {Jeff Erickson},
  title        = {An optimal generalization of the centerpoint theorem, and its extensions},
  booktitle    = {Proceedings of the 23rd {ACM} Symposium on Computational Geometry ({SCG}'07),
                  Gyeongju, South Korea, June 6-8, 2007},
  pages        = {138--141},
  publisher    = {{ACM}},
  year         = {2007},
  url          = {https://doi.org/10.1145/1247069.1247097},
  doi          = {10.1145/1247069.1247097},
  bibsource    = {dblp computer science bibliography, https://dblp.org}
}

@book {rockafellarConvexAnalysis1970,
    AUTHOR = {Rockafellar, R. Tyrrell},
     TITLE = {Convex analysis},
    SERIES = {Princeton Mathematical Series},
    VOLUME = {No. 28},
 PUBLISHER = {Princeton University Press, Princeton, NJ},
      YEAR = {1970},
     PAGES = {xviii+451},
   MRCLASS = {26.52 (46.00)},
  MRNUMBER = {274683},
  doi = {10.1515/9781400873173},
  isbn = {978-1-4008-7317-3}
  
}

@article{sikorskiEllipsoidAlgorithmComputation1993,
  author       = {Christopher (Krzysztof) Sikorski and
                  Chey{-}Woei Tsay and
                  Henryk Wozniakowski},
  title        = {An Ellipsoid Algorithm for the Computation of Fixed Points},
  journal      = {J. Complex.},
  volume       = {9},
  number       = {1},
  pages        = {181--200},
  year         = {1993},
  url          = {https://doi.org/10.1006/jcom.1993.1013},
  doi          = {10.1006/JCOM.1993.1013},
  bibsource    = {dblp computer science bibliography, https://dblp.org}
}

@article {sikorskiComputationalComplexityFixed2009,
    AUTHOR = {Sikorski, Krzysztof},
     TITLE = {Computational complexity of fixed points},
   JOURNAL = {J. Fixed Point Theory Appl.},
  FJOURNAL = {Journal of Fixed Point Theory and Applications},
    VOLUME = {6},
      YEAR = {2009},
    NUMBER = {2},
     PAGES = {249--283},
      ISSN = {1661-7738,1661-7746},
   MRCLASS = {47H10 (68Q25)},
  MRNUMBER = {2580978},
       DOI = {10.1007/s11784-009-0128-3},
       URL = {https://doi.org/10.1007/s11784-009-0128-3},
}

@article{vapnik1971vcdimension,
  title = {On the Uniform Convergence of Relative Frequencies of Events to Their Probabilities},
  author = {Vapnik, Vladimir N. and Chervonenkis, Alexey Y.},
  year = {1971},
  journal = {Theory of Probability \& Its Applications},
  volume = {16},
  number = {2},
  pages = {264--280},
  doi = {10.1137/1116025}
}

@incollection{zivaljevic2017handbookchapter,
  title = {Topological Methods in Discrete Geometry},
  booktitle = {Handbook of Discrete and Computational Geometry},
  author = {{\v Z}ivaljevi{\'c}, Rade T.},
  year = {2017},
  edition = {3},
  pages = {551--580},
  publisher = {CRC Press LLC},
  chapter = {21}
}


\newpage
\appendix

\section{More on \texorpdfstring{$\ell_p$}{lp}-Halfspaces}\label{sec:appendixProofsStructure}

In this section, we will study limit $\ell_p$-halfspaces in more detail. The main goal is to provide proofs for the halfspace properties in \Cref{ssec:properties} and \Cref{ssec:discrete_algorithms}, but we will also develop some additional theory on the way. We split the exposition into two parts: fundamental properties of $\ell_p$-halfspaces and insights about their shape are developed in \Cref{ssec:properties_I} while \Cref{ssec:properties_II} focuses on the interaction of $\ell_p$-halfspaces with mass distributions.

\subsection{Fundamentals of \texorpdfstring{$\ell_p$}{lp}-Halfspaces}
\label{ssec:properties_I}

For completeness' sake, we recall the definition of limit $\ell_p$-halfspaces, which we will just refer to as $\ell_p$-halfspaces in this appendix. 

\limithalfspace*

Recall that $\limitH^p_{x, v}$ is invariant under scaling $v$ with a positive scalar, and we therefore usually use $v \in S^{d - 1}$. As a warm-up, we prove \Cref{lemma:axis_aligned_halfspace}, which considers $\ell_p$-halfspaces with axis-aligned directions $v$.

\lemmaaxisalignedhalfspace*
\begin{proof}
    Without loss of generality, assume that $v = (1, 0, \dots, 0) \in S^{d - 1}$. By our definition of $\ell_p$-halfspaces, $z$ is contained in $\limitH^p_{x, v}$ if and only if $\norm{x - z}_p^p \leq \norm{x - \epsilon v - z}_p^p$ for all $\epsilon > 0$, which happens if and only if $|x_1 - z_1|^p \leq |x_1 - \epsilon - z_1|^p$. This last inequality holds if and only if $z_1 \geq x_1$, independently of $p$. This concludes the proof.
\end{proof}

The proofs of the other properties require some more effort. In particular, for \Cref{lem:structureoflimithalfspace}, we will need \Cref{obs:characterization_containment}, which we restate here for convenience.

\obscharacterizationcontainment*

With this, we are ready to prove \Cref{lem:structureoflimithalfspace}, which states that $\ell_p$-halfspaces are unions of rays originating at $x$. The proof mainly relies on convexity of $\ell_p$-balls and \Cref{obs:characterization_containment}.

\lemstructureoflimithalfspace*
\begin{proof}
    We start with the first part of the statement for arbitrary $p \in [1, \infty) \cup \{\infty\}$. For this, let $L$ be the line through $x$ in direction $v$, and let $R_-$ be the open ray from $x$ in direction $-v$. Let $B_z$ be the smallest $\ell_p$-ball around a point $z \in \R^d$ that contains $x$. From \Cref{obs:characterization_containment}, we know that 
    \[
        z \in \limitH^p_{x, v} \iff R_- \cap B_z^\circ = \emptyset, 
    \]
    where $B_z^\circ$ denotes the interior of $B_z$. 
    
    Consider what happens when we move $z$ along the ray through $x$ and $z$. In other words, consider a point $z' = x + \delta (z - x)$ for some $\delta > 0$. 
    By definition, $x$ is on the boundary of both balls $B_{z'}$ and $B_z$. Moreover, observe that the ball $B_{z'}$ is a translated (with center $z'$ instead of $z$) and scaled (such that $x$ stays on the boundary) version of $B_z$. In particular, both $B_z$ and $B_{z'}$ have the same tangent hyperplanes at the point $x$. Thus, we conclude 
    \[
        R_- \cap B_z^\circ = \emptyset \iff R_- \cap B_{z'}^\circ = \emptyset,
    \]
    which means that $z' \in \limitH^p_{x, v}$ if and only if $z \in \limitH^p_{x, v}$. This proves that $\limitH^p_{x, v}$ is a union of rays starting at $x$. 

    It remains to prove that the boundary $\partial \limitH^p_{x, v}$ of $\limitH^p_{x, v}$ is a union of lines through $x$, assuming that $p \in (1, \infty)$. For this, observe that the $\ell_p$-balls for $p \in (1, \infty)$ are strictly convex. We claim that a point $z$ is on the boundary of $\limitH^p_{x, v}$ if and only if $B^\circ_z \cap L = \emptyset$. Indeed, any point $z$ with $B^\circ_z \cap L = \emptyset$ is clearly contained in $\limitH^p_{x, v}$ (by \Cref{obs:characterization_containment}). Moreover, moving $z$ into the direction $-v$ yields a point $z'$ with $B^\circ_{z'} \cap R_- \neq \emptyset$ (by strict convexity of $\ell_p$-balls). We conclude $z \in \partial \limitH^p_{x, v}$. Conversely, consider an arbitrary point $z$ with $B^\circ_z \cap L \neq \emptyset$. If $B^\circ_z$ intersects $R_-$, slightly moving $z$ does not change this. Concretely, any point $z'$ in a small neighborhood around $z$ satisfies $B^\circ_{z'} \cap R_- \neq \emptyset$, and thus none of these points is contained in the halfspace. We conclude that $z$ cannot be on the boundary. Otherwise, $B^\circ_z$ intersects $R_+ \coloneqq L \setminus (R_- \cup \{x\})$, but then slightly moving $z$ does not change this. In other words, any point $z'$ in a small neighborhood around $z$ satisfies $B^\circ_{z'} \cap R_+ \neq \emptyset$, which also implies $B^\circ_{z'} \cap R_- = \emptyset$ (as otherwise, $x$ would be contained in $B^\circ_{z'}$ by convexity). We conclude that $z$ is not on the boundary.

    With this characterization of the boundary points, we can now finish the proof. Consider an arbitrary $z \in \partial \limitH^p_{x, v}$. We must have $B^\circ_z \cap L = \emptyset$. In other words, $L$ is tangential to $B_z$. This remains true for any other point $z' = x + \delta (z - x)$ for $\delta \in \R$ on the line through $x$ and $z$ with ball $B_{z'}$ (using point-symmetry of $\ell_p$-balls around their center for $\delta < 0$). This again implies $z' \in \partial \limitH^p_{x, v}$.
\end{proof}

With \Cref{lem:structureoflimithalfspace} at our disposal, we now want to tackle \Cref{lem:anglesonlimithalfspaces} next. However, before getting there, we will need to recall some tools from convex analysis. Concretely, we will need the notion of subgradients of a convex function $f : \R^d \rightarrow \R$: a vector $u \in \R^d$ is a subgradient of $f$ at $x \in \R^d$ if and only if 
\[
    f(x') - f(x) \geq \langle u, (x' - x) \rangle
\]
for all $x' \in \R^d$. The set $\partial f(x) \subseteq \R^d$ of all subgradients of $f$ at $x$ is also called the subdifferential. If $f$ is differentiable at $x$, then $\partial f(x) = \{\nabla f(x)\}$. For more details on convex analysis, we refer to the standard textbook by Rockafellar~\cite{rockafellarConvexAnalysis1970}.

Subgradients are useful because they allow us to further characterize containment of a given point $z$ in an $\ell_p$-halfspace $\limitH^p_{x, v}$ as follows.

\lemmasubgradientscontainment*
\begin{proof}
    Assume first that there exists such a subgradient $u \in \partial ||z - x||_p$ with $\langle u, v \rangle \geq 0$. Choosing $x' = z - x + \epsilon v$ in the definition of subgradients, we conclude $||z - x + \epsilon v||_p - ||z - x||_p \geq \epsilon \langle u, v \rangle \geq 0$ for all $\epsilon > 0$. Thus, $z$ is contained in the halfspace. Conversely, assume $z \in \limitH^p_{x, v}$. By \Cref{obs:characterization_containment}, this implies that the intersection of the open ray $R_-$ from $x$ in direction $-v$ with the interior $B^\circ_z$ of the smallest ball $B_z$ containing $x$ around $z$ is empty. Both of these objects are convex. Thus, there exists a hyperplane separating $R_-$ and $B^\circ_z$. Note that this hyperplane must go through $x$, and a scaling of its normal vector must be a subgradient in $\partial ||x - z||_p$. Naturally, this subgradient $u$ at $x$ points away from $B^\circ_z$, and we have $\langle u, -v \rangle \geq 0$. It remains to observe $-u \in \partial ||z - x||_p$ ($-u$ is the desired subgradient).
\end{proof}

In order to use \Cref{lemma:subgradients_and_containment}, we need to determine the subgradients of the $\ell_p$-norm for all $p \in [1, \infty) \cup \{\infty\}$. We start with the cases $p \in (1, \infty)$, and focus on those points $z \in \R^d$ with $||z||_p = 1$. Note that by \Cref{lem:structureoflimithalfspace}, knowledge about the cases with $||z||_p = 1$ already suffices to make full use of \Cref{lemma:subgradients_and_containment}.

\begin{observation}
\label{obs:differentiable_cases}
    For $p \in (1, \infty)$ and arbitrary $z \in \R^d$ with $||z||_p = 1$, $||\cdot||_p$ is differentiable at $z$ and the gradient $\nabla ||z||_p$ is given by
    \[
        \nabla ||z||_p=\begin{bmatrix}
            |z_1|^{p-1} sign(z_1) \\
            \vdots \\
            |z_d|^{p-1} sign(z_d)
        \end{bmatrix}.
    \]
    In particular, we have $\langle \nabla ||z||_p, z \rangle = \left\lVert  \nabla \left( ||z||_p  \right) \right\rVert_{\frac{p}{p - 1}} =||z||_p = 1$.
\end{observation}
Next, we characterize the subgradients for $||\cdot||_\infty$ and $||\cdot||_1$ on points $z \in \R^d$ with $||z||_\infty = 1$ ($||z||_1 = 1$, respectively). Note that these subdifferentials can also be found in the literature (see, e.g.,~\cite{rockafellarConvexAnalysis1970}). 
\begin{observation}
\label{obs:l_infty_subdifferential}
    Consider arbitrary $z \in \R^d$ with $||z||_\infty = 1$. The set of subgradients $\partial ||z||_\infty$ is given by
    \[
        \partial ||z||_\infty= \conv \left( \{-e_i \mid z_i = -1 \} \cup \{e_i \mid z_i = 1\} \right), 
    \]
    where $e_i$ denotes the $i$-th standard unit vector in $\R^d$. In particular, we have $\langle u, z \rangle = ||u||_1 = ||z||_\infty = 1$ for all $u \in \partial ||z||_\infty$.
\end{observation}
\obssubdifferentialmanhattan*

Consider a subgradient $u$ of $||\cdot ||_p$ at point $z$. Given the characterization of $\ell_p$-halfspaces in \Cref{lemma:subgradients_and_containment}, it will be useful to prove an upper bound on the angle between $u$ and $z$.
\begin{lemma}
\label{lemma:bound_angle_subgradients}
    Let $p \in [1, \infty) \cup \{ \infty\}$ and consider arbitrary $z \in \R^d$ with $||z||_p = 1$. Then we have $\measuredangle(u, z) \leq \frac{\pi}{2} - \sqrt{\nicefrac{1}{d}}$ for all subgradients $u \in \partial ||z||_p$.
\end{lemma}
\begin{proof}
    We heavily use the previous observations. In particular, recall that we have $\langle u, z \rangle = ||u||_{\frac{p}{p - 1}} = ||z||_p = 1$ for all subgradients $u \in \partial ||z||_p$ (where $\frac{p}{p - 1}$ is interpreted as $1$ for $p = \infty$ and as $\infty$ for $p = 1$, respectively). We thus get
    \[
    \cos(\measuredangle(u, z)) 
    = \frac{ \langle u,  z \rangle }{||u||_2||z||_2} 
    = \frac{1}{||u||_2||z||_2 }.
    \]
    Recall the $\ell_p$-norm inequalities $||x||_2 \leq ||x||_p$ for $p \leq 2$ and  $||x||_2 \leq \sqrt{d}||x||_p$ for $2 \leq p$. Further, for any $p \in [1, \infty) \cup \{\infty\}$ we have $p < 2 \iff \frac{p}{p - 1} > 2$. Thus, we obtain
    \[
        \frac{1}{||u||_2||z||_2 } \geq \frac{1}{ \sqrt{d} ||u||_{\frac{p}{p - 1}} ||z||_p }
        = \frac{1}{\sqrt{d}}.
    \]
    Thus, we conclude $\measuredangle(u, z) \leq \frac{\pi}{2} - \frac{1}{\sqrt{d}}$, as desired.
\end{proof}

With this theory in place, we are now ready to prove \Cref{lem:anglesonlimithalfspaces}. 

\lemanglesonlimithalfspaces*
\begin{proof}
    By \Cref{lem:structureoflimithalfspace} (every halfspace is a union of rays), we can assume $||z - x||_p = 1$, without loss of generality. Now consider the first statement of the lemma. \Cref{lemma:subgradients_and_containment} says that $z \in \limitH^p_{x, v}$ implies $\langle u, v \rangle \geq 0$ for some subgradient $u \in \partial ||z - x||_p$. Concretely, this means that $\measuredangle (u, v) \leq \frac{\pi}{2}$. Moreover, by \Cref{lemma:bound_angle_subgradients}, we have $\measuredangle (u, z - x) \leq \frac{\pi}{2} - \sqrt{\nicefrac{1}{d}}$. Combining the two bounds, we conclude $\measuredangle(\overrightarrow{xz}, v) = \measuredangle(z - x, v) \leq \pi - \sqrt{\nicefrac{1}{d}}$, as desired.

   For the second statement, assume that $z$ satisfies $\measuredangle(\overrightarrow{xz}, v) \leq \sqrt{\nicefrac{1}{d}}$.
   By \Cref{lemma:bound_angle_subgradients}, there exists a subgradient $u \in \partial ||z - x||_p$ with $\measuredangle (u, z - x) \leq \frac{\pi}{2} - \sqrt{\nicefrac{1}{d}}$. We conclude $\measuredangle(u, v) \leq \pi$, which implies $z \in \limitH^p_{x, v}$ by \Cref{lemma:subgradients_and_containment}.
\end{proof}

Finally, we can prove \Cref{lemma:pull_towards_zero} using \Cref{lem:anglesonlimithalfspaces}.

\lemmapulltowardszero*
\begin{proof}
    By \Cref{lem:anglesonlimithalfspaces}, it suffices to prove that all points $z \in [0, 1]^d$ satisfy $\measuredangle(\overrightarrow{xz}, -x) \leq \sqrt{\nicefrac{1}{d}}$. 
    To see this, consider the triangle spanned by the origin $0 \in \R^d$ and the two points $z$ and $x$. We know that $||z-0||_2 = ||z||_2 \leq \sqrt{d}$, and $||x-0||_2= ||x||_2 > 2d$. By the law of sines, we get
    \begin{align*}
        \sin (\measuredangle(\overrightarrow{xz}, \overrightarrow{x0})) &= ||z-0||_2  \frac{\sin(\measuredangle(\overrightarrow{zx},\overrightarrow{z0}))}{||x-0||_2} \\
        &\leq \sqrt{d} \frac{\sin(\measuredangle(\overrightarrow{zx},\overrightarrow{z0}))}{||x-0||_2} \\
        &<  \sqrt{d}\frac{1}{2d} \\
        &= \frac{1}{2\sqrt{d}}
    \end{align*}
    which implies $\measuredangle(\overrightarrow{xz}, \overrightarrow{x0}) < \sqrt{\nicefrac{1}{d}}$.
\end{proof}

\subsection{\texorpdfstring{$\ell_p$}{lp}-Halfspaces and Mass Distributions}
\label{ssec:properties_II}

We now move on to the interaction of $\ell_p$-halfspaces with mass distributions. We briefly recall that mass distributions are finite absolutely continuous (w.r.t.\ the Lebesgue measure) measures on $\R^d$. As mentioned in \Cref{ssec:properties}, by the Radon-Nikodym theorem, we can therefore think of mass distributions as probability distributions with a density (by assuming $\mu(\R^d) = 1$, without loss of generality). We will use this point of view repeatedly in this section: it allows us to use the language of probability theory (instead of the language of measure theory).

Before we prove \Cref{lemma:continuity_halfspace_mass}, i.e., that $\mu(\limitH_{x,v}^p)$ is continuous in $x$, we first want to prove that the boundary of an $\ell_p$-halfspace has measure $0$. The following lemma will help with this.

\begin{lemma}[Inside and Outside Orthant]
\label{lemma:inside_outside_orthant}
    Let $\limitH^p_{x, v}$ be an arbitrary $\ell_p$-halfspace for some $p \in [1, \infty) \cup \{\infty\}$. Consider arbitrary $z, z' \in \R^d$ satisfying 
    both 
    \begin{itemize}
        \item $v_i > 0 \implies z'_i - z_i \geq 0$
        \item and $v_i < 0 \implies z'_i - z_i \leq 0$ 
    \end{itemize} 
    for all $i \in [d]$. Then  $z \in \limitH^p_{x, v}$ implies $z' \in \limitH^p_{x, v}$. Similarly, for any $z \notin \limitH^p_{x, v}$, all $z' \in \R^d$ satisfying $v_i > 0 \implies z'_i - z_i \leq 0$ and $v_i < 0 \implies z'_i - z_i \geq 0$ for all $i \in [d]$ are not contained in $\limitH^p_{x, v}$.
\end{lemma}
\begin{proof}
    Let $z \in \limitH^p_{x, v}$ be arbitrary and consider first the case $p \neq \infty$. 
    Without loss of generality, we can consider the case where $z'$ and $z$ only differ in one coordinate, i.e., $z'_i = z_i + \alpha v_i$ for some $i \in [d]$ and $\alpha > 0$, and $z'_j = z_j$ for all $j \in [d] \setminus \{i\}$. By definition, we must have 
    \[
        \norm{x - z}^p_p \leq \norm{x - \epsilon v - z}^p_p
    \]
    for all $\epsilon > 0$. If we can prove that 
    \begin{equation}\label{ineq:coordinatebetter}
        |x_i - \epsilon v_i - z_i|^p - |x_i - z_i|^p \leq |x_i - \epsilon v_i - z'_i|^p - |x_i - z'_i|^p
    \end{equation}
    for all $\epsilon > 0$, we are done, because then 
    \[
    0 \leq \norm{x - \epsilon v - z}^p_p - \norm{x - z}^p_p \leq \norm{x - \epsilon v - z'}^p_p - \norm{x - z'}^p_p
    \]
    holds for all $\epsilon > 0$ and thus $z' \in \limitH^p_{x, v}$. \Cref{ineq:coordinatebetter} follows from the observation that
    \[
        |x_i - \epsilon v_i - z'_i|^p - |x_i - z'_i|^p = |x_i - (\alpha + \epsilon) v_i - z_i|^p - |x_i - \alpha v_i - z_i|^p
    \]
    and that the function 
    $f(\beta) \coloneqq |\beta + \gamma|^p - |\beta|^p$ is for all $p\in[1,\infty)$ monotonically non-increasing for $\gamma \leq 0$ and monotonically non-decreasing for $\gamma \geq 0$. The proof for $p=\infty$ follows from $p=1$, since the $\ell_\infty$-norm is just taking the maximum over the constituents of the $\ell_1$-norm.
    
    Consider now $z \notin \limitH^p_{x, v}$ and consider again first the case $p \neq \infty$. As before, we can assume without loss of generality that $z'_i = z_i - \alpha v_i$ for some $i \in [d]$ and $\alpha > 0$, and $z'_j = z_j$ for all $j \in [d] \setminus \{i\}$. We want to prove $z' \notin \limitH^p_{x, v}$. By $z \notin \limitH^p_{x, v}$, there must exist $\epsilon > 0$ such that
    \[
        \norm{x - z}^p_p > \norm{x - \epsilon v - z}^p_p.
    \]
    If we can prove 
    \begin{equation}\label{ineq:coordinateTwo}
        |x_i - z_i|^p - |x_i - \epsilon v_i - z_i|^p \leq |x_i - z'_i|^p - |x_i - \epsilon v_i - z'_i|^p,
    \end{equation}
    we get 
    \[
    0 < \norm{x - z}^p_p   - \norm{x - \epsilon v - z}^p_p  \leq \norm{x - z'}^p_p - \norm{x - \epsilon v - z'}^p_p 
    \]
    and thus $z' \notin \limitH^p_{x, v}$. The inequality \Cref{ineq:coordinateTwo}
    follows from the observation that
    \[
        |x_i - z'_i|^p - |x_i - \epsilon v_i - z'_i|^p = |x_i + \alpha v_i - z_i|^p - |x_i + (\alpha - \epsilon) v_i - z_i|^p
    \]
    and that the function 
    $f(\beta) \coloneqq |\beta|^p - |\beta - \gamma|^p$  is for all $p\in[1,\infty)$ monotonically non-increasing for $\gamma \leq 0$ and monotonically non-decreasing for $\gamma \geq 0$. The proof for $p=\infty$ again follows from $p=1$.
\end{proof}

This now lets us conclude that the boundary of the halfspace has measure $0$.
\begin{corollary}[Boundary has Measure $0$]
\label{corollary:boundary_measure_zero}
    Let $p \in [1, \infty) \cup \{\infty\}$. The boundary $\partial \limitH^p_{x, v}$ of the $\ell_p$-halfspace $\limitH^p_{x, v}$ has Lebesgue measure $0$ for all $x \in \R^d$ and $v \in S^{d - 1}$. 
\end{corollary}
\begin{proof}
    We prove that $\partial \limitH^p_{x, v}$ is porous, i.e., we prove that there exists $0 < \alpha < 1$ such that for every sufficiently small $r > 0$ and for every $z \in \R^d$, there is $y \in \R^d$ such that $B^p(y, \alpha r) \subseteq B^p(z, r) \setminus \partial \limitH^p_{x, v}$. This is a sufficient condition for a subset of $\R^d$ to have measure $0$.
    Porosity follows from \Cref{lemma:inside_outside_orthant} by distinguishing between the cases $z \in \limitH^p_{x, v}$ and $z \notin \limitH^p_{x, v}$.
\end{proof}

We are now ready to prove \Cref{lemma:continuity_halfspace_mass}. While this could be proven purely in terms of measure theory, we opted for the language of probability theory instead.

\lemmacontinuityhalfspacemass*
\begin{proof}
     Fix a mass distribution $\mu$, and fix $v\in S^{d-1}$ and $p\in[1,\infty)\cup\{\infty\}$. Without loss of generality, assume $\mu(\R^d) = 1$. Let $x\in \R^d$ be arbitrary, and consider an arbitrary sequence $(x_n)_{n \in \N}$ converging to $x$. For $n \in \N$, let $X_n$ denote the indicator random variable for a point drawn at random (w.r.t.\ $\mu$) to lie in the set $\limitH^p_{x_n, v}$. Similarly, let $X$ denote the indicator random variable for containment in $\limitH^p_{x, v}$. We claim that $X_n$ converges almost surely to $X$ as $n$ goes to infinity. To see this, observe first that we naturally have 
     \[
        \lim_{n \rightarrow \infty} \norm{z - x_n}_p = \norm{z - x}_p
    \]
     for all $z \in \R^d$. Now let $z \notin \limitH^p_{x, v}$ be arbitrary (i.e., $X(z) = 0$). Then there exist $\epsilon, \delta > 0$ such that $\norm{x - \epsilon v - z}_p < \norm{x - z}_p - \delta$. Thus, for $n \in \N$ large enough we also have $X_n(z) = 0$. Consider now an arbitrary $z$ in the interior of $\limitH^p_{x, v}$. By definition, there exists $\delta > 0$ such that a $\delta$-ball around $z$ is contained in $\limitH^p_{x, v}$. In other words, we have 
     \[
        \norm{x - \epsilon v - z'}_p \geq \norm{x - z'}_p
     \]
     for all $\epsilon > 0$ and $z'$ with $\norm{z' - z}_p \leq \delta$. For large enough $n \in \N$, we can therefore choose $z'$ such that $x - z' = x_n - z$, which yields
     \[
        \norm{x_n - \epsilon v - z}_p \geq \norm{x_n - z}_p
     \]
     and thus $\lim_{n \rightarrow \infty} X_n(z) = 1$. We conclude that the set $\{z \in \R^d \mid \lim_{n \rightarrow \infty} X_n(z) = X(z) \}$ includes all of $\R^d$ except for the boundary of $\limitH^p_{x, v}$. By \Cref{corollary:boundary_measure_zero},
     this boundary has measure zero and we get almost sure convergence. Using the dominated convergence theorem, this implies convergence in mean and thus
     \[
        \lim_{n \rightarrow \infty} \mu(\limitH^p_{x_n, v}) = \lim_{n \rightarrow \infty} \mathbb{E}(X_n) = \mathbb{E}(X) = \mu(\limitH^p_{x, v}),
     \]
     as desired.
\end{proof}

While we cannot prove for all $p$ that $\mu(\limitH_{x,v}^p)$ is continuous in $v$, we at least prove \Cref{lemma:pulling_directions_positive_measure}, which states that the set $V$ of directions $v$ for which $\mu(\limitH^p_{x,-v})$ is strictly smaller than some threshold $t$, is open.

\lemmapullingdirectionspositivemeasure*
\begin{proof}
    Fix $\mu$ and $p \in [1, \infty) \cup \{\infty\}$ and let $x \in \R^d$ be arbitrary. Assume without loss of generality that $\mu(\R^d)$. Instead of directly proving that $V$ is open, we will prove that its complement
    \[ 
        S^{d - 1} \setminus V = \{ v \in S^{d - 1} \mid \mu(\limitH^p_{x, v}) \geq t \}
    \]

    is closed. To this end, consider an arbitrary sequence $(v_n)_{n \in \N}$ converging to $v \in S^{d - 1}$ satisfying $v_n \in S^{d - 1} \setminus V $ for all $n \in \N$. We will prove that then $v \in S^{d - 1} \setminus V $. Let $X_n$ be the indicator random variable for containment of a point drawn according to $\mu$ in $\limitH^p_{x, v_n}$. Observe that by $\lim_{n \rightarrow \infty} \norm{x - \epsilon v_n - z}_p = \norm{x - \epsilon v - z}_p$ for all $z \in \R^d$, the sequence of random variables $(X_n)_{n \in \N}$ converges point-wise to a random variable $X$. In other words, $(X_n)_{n \in \N}$ converges almost surely to $X$, and by the dominated convergence theorem we get convergence in mean. Concretely, this means that we have 
    \[
        \lim_{n \rightarrow \infty} \mu(\limitH^p_{x, v_n}) = \mathbb{E} X \geq  t.
    \]
    Let now $z \notin \limitH^p_{x, v}$ be arbitrary. Then there exist $\epsilon, \delta > 0$ such that $\norm{x - \epsilon v - z}_p < \norm{x - z}_p - \delta$. Using $\lim_{n \rightarrow \infty} \norm{x - \epsilon v_n - z}_p = \norm{x - \epsilon v - z}_p$, we conclude that $X(z) = 0$. In other words, we just proved that $\R^d \setminus \limitH^p_{x, v} \subseteq \{ z \in \R^d \mid X(z) = 0\}$ and thus $\mu(\limitH^p_{x, v}) \geq \mathbb{E} X$. Plugging everything together, we obtain $\mu(\limitH^p_{x, v}) \geq t$ and hence $v \in S^{d - 1} \setminus V$, as desired.
\end{proof}

\section{Proof of \texorpdfstring{\Cref{theorem:discrete_p-centerpoint}}{Theorem 3.18}}\label{sec:appendixProofCenterpoint}

\theoremdiscretepcenterpoint*
\begin{proof}
    Without loss of generality, assume that $P \subseteq (0, 1)^d$. Let $r > 0$ be arbitrary and consider the mass distribution $\mu_r$ on $\R^d$ obtained by replacing each point $z$ in $P$ with an $\ell_p$-norm ball $B^p(z, r)$ of radius $r$ around $z$ (without loss of generality, assume $r$ is small enough such that the balls do not intersect and such that the support of $\mu_r$ is contained in $[0, 1]^d$). \Cref{theorem:p-centerpoint} yields a centerpoint $c^{(r)}$ of $\mu_r$. Moreover, by \Cref{obs:centerpoint_in_bounding_box}, these centerpoints are contained in $[0, 1]^d$ as well. Now consider the sequence of centerpoints $c^{(r)}$ as $r$ goes to $0$. By the theorem of Bolzano-Weierstrass, this sequence has a convergent subsequence with limit $c \in [0, 1]^d$. From now on, we will use $c^{(r)}$ to refer to elements of this convergent subsequence.

    We claim that $c$ is an $\ell_p$-centerpoint of $P$. To prove this, fix $v \in S^{d - 1}$ and let $z \in P \setminus \limitH^p_{c, v}$ be arbitrary. This means that there exists $\epsilon > 0$ such that $\norm{c - z}_p > \norm{c - \epsilon v - z}_p$ and thus even $\norm{c - z}_p > \norm{c - \epsilon v - z}_p + \delta$ for some $\delta > 0$. By choosing $r$ small enough, we can ensure both $r \leq \delta / 4$ and $\norm{c^{(r)} - c}_p \leq \delta / 4$. Using the triangle inequality, we then get 
    \begin{align*}
        \norm{c^{(r)} - z'}_p 
        &\geq \norm{c - z}_p - \norm{z' - z}_p - \norm{c - c^{(r)}}_p \\
        &\geq \norm{c - z}_p - \delta / 2 \\ 
        &> \norm{c - \epsilon v - z}_p + \delta / 2 \\ 
        &\geq \norm{c^{(r)} - \epsilon v - z'}_p - \norm{z' - z}_p - \norm{c - c^{(r)}}_p + \delta / 2 \\
        &\geq \norm{c^{(r)} - \epsilon v - z'}_p
    \end{align*}
    for all $z' \in B^p(z, r)$. We conclude that $ B^p(z, r) \cap \limitH^p_{c^{(r)}, v}$ is empty. In fact, this derivation holds true for all sufficiently small $r$. This means that all balls around points $z \in P \setminus \limitH^p_{c, v}$ will eventually (for small enough $r$) stop contributing to $\mu_r(\limitH^p_{c^{(r)}, v})$. By $\mu_r(\limitH^p_{c^{(r)}, v}) \geq \frac{1}{d + 1} \mu_r(\R^d)$, we conclude that there can be at most $\frac{d}{d + 1}|P|$ such points $z$ and thus we get $|P \cap \limitH^p_{c, v}| \geq \frac{|P|}{d + 1}$, as desired.
\end{proof}

\end{document}